\documentclass{article} 
\usepackage[margin=1.1in]{geometry}
\usepackage{authblk}

\usepackage[round]{natbib}
\usepackage[utf8]{inputenc} 
\usepackage[T1]{fontenc}    
\usepackage{times}
\usepackage{custom}
\usepackage{algorithm}

\title{Synthetic-Powered Multiple Testing with FDR Control}

\author[1]{Yonghoon Lee$^*$}
\author[2]{Meshi Bashari$^*$}
\author[1]{Edgar Dobriban}
\author[2,3]{Yaniv Romano}
\affil[1]{Department of Statistics and Data Science, The Wharton School, University of Pennsylvania, USA}
\affil[2]{Department of Electrical and Computer Engineering, Technion IIT, Israel}
\affil[3]{Department of Computer Science, Technion IIT, Israel}
\date{}

\begin{document}
\maketitle
\begingroup
\renewcommand\thefootnote{\fnsymbol{footnote}}
  \setcounter{footnote}{1}
  \footnotetext{Equal contribution.}
  \setcounter{footnote}{0}
\endgroup
\begin{abstract}
Multiple hypothesis testing with false discovery rate (FDR) control is a fundamental problem in statistical inference, with broad applications in genomics, drug screening, and outlier detection. In many such settings, researchers may have access not only to real experimental observations but also to auxiliary or synthetic data---from past, related experiments or generated by generative models---that can provide additional evidence about the hypotheses of interest.
We introduce \texttt{SynthBH}, a synthetic-powered multiple testing procedure that safely leverages such synthetic data.
We prove that \texttt{SynthBH} guarantees finite-sample, distribution-free FDR control under a mild PRDS-type positive dependence condition, without requiring the pooled-data p-values to be valid under the null.
The proposed method adapts to the (unknown) quality of the synthetic data: it enhances the sample efficiency and may boost the power when synthetic data are of high quality, while controlling the FDR at a user-specified level regardless of their quality.
We demonstrate the empirical performance of \texttt{SynthBH} on tabular outlier detection benchmarks and on genomic analyses of drug\mbox{--}cancer sensitivity associations, and further study its properties through controlled experiments on simulated data.
\end{abstract}

\section{Introduction}

Multiple hypothesis testing is a cornerstone of modern statistical inference.
In large-scale scientific studies---including genomics, drug discovery, high-throughput screening, and large-scale anomaly detection---researchers routinely test thousands to millions of hypotheses simultaneously.
In such regimes, controlling the \emph{false discovery rate} (FDR) \citep{benjamini1995controlling} has become a default target, as it offers a favorable balance between statistical validity and power compared to more stringent criteria such as family-wise error control.

A persistent bottleneck in these applications is that the amount of \emph{trusted} real data is often limited.
For example, in genomics, the number of reliably measured samples can be small relative to the dimensionality; in drug--cancer sensitivity studies, each additional experiment can be costly; and in outlier detection, obtaining a clean reference set of inliers can be expensive or require manual verification.
At the same time, practitioners increasingly have access to large amounts of \emph{auxiliary} data that are not fully trustworthy but are often informative: past experiments on related populations, weakly labeled or automatically curated datasets, and synthetic samples generated by modern generative models.
Such auxiliary datasets can be far larger than the real dataset, and when they are high quality, they have the potential to  dramatically sharpen statistical evidence.
However, because the auxiliary data distribution can differ from the real one in unknown ways, naively pooling real and synthetic data in classical testing pipelines can lead to spurious discoveries and inflated FDR.

This tension creates a basic methodological gap.
On the one hand, ignoring synthetic data can be overly conservative and low-powered in small-sample regimes.
On the other hand, treating synthetic data as if they were real can destroy the very error guarantees that make multiple testing scientifically reliable.
This paper asks a concrete question:
\emph{Can we leverage arbitrary synthetic/auxiliary data to improve the power of multiple testing, while guaranteeing finite-sample FDR control that is robust to unknown synthetic data quality?}

\section{Synthetic-Powered P-Values}
{\bf Setting.}
We consider a standard multiple hypothesis testing framework with $m$ hypotheses $H_1,\ldots,H_m$.
 Each of hypothesis refers to a null claim---for example, $H_j$ referring to ``the presence of the $j$-th genomic feature is not associated with the response to a given drug.'' 
For each hypothesis $H_j$, we assume access to (i) a \emph{valid} p-value $p_j$ computed from the trusted real data only,  and (ii) an additional p-value $\tilde p_j$ computed from the merged real-and-synthetic dataset.
The merged p-value $\tilde p_j$ can be substantially more informative when the synthetic data are of high quality (since it effectively uses a larger sample), but it is generally \emph{not} guaranteed to be valid under the null because the synthetic data distribution may be arbitrary.

Our contributions in this work begin with the formulation of a \emph{synthetic-powered p-value}.
Let $a\wedge b:=\min(a,b)$ and 
$a\vee b:=\max(a,b)$ denote the minimum and maximum of two numbers, respectively.
The \textit{synthetic-powered p-value at level $\delta\ge0$} is defined as
\begin{equation}\label{eqn:spp}
    \tilde{p}_j^\delta = p_j \wedge (\tilde{p}_j \vee (p_j - \delta)).
\end{equation}
This construction is deliberately careful in using synthetic data. Suppose we are handed high-quality synthetic data. If they provide very strong evidence, so that $\tilde{p}_j$ is smaller than $p_j$, we want to use it instead of the real data p-value $p_j$; and we take $ p_j \wedge \tilde{p}_j$ to achieve this. 
At the same time, since the quality of the synthetic data is unknown in general, we account for the possibility that they may be poor or misleading, and we cap the effect of the synthetic p-value at $p_j - \delta$. 

Generally, $\tilde{p}_j^\delta$ is not guaranteed to be a super-uniform variable---so not technically a classical p-value---but we use the term ``synthetic-powered p-value'' with perhaps a slight abuse of notation, because we will use it as an input to multiple testing procedures.\footnote{We show that the final procedure controls the false discovery rate; thus the terminology of $p$-values is still justified.}

{\bf From p-values to FDR control: \texttt{SynthBH}.}
The main technical contribution of this paper is to show how to turn the guarded p-values $\tilde p_j^\delta$, for appropriate choices of $\delta\ge 0$,
into a multiple testing procedure with provable FDR control.
We consider 
a user-specified \emph{admission cost} $\ep\ge 0$ for incorporating synthetic data: larger $\eps$ allows potentially larger power gains, 
but can also potentially increase false discoveries.

Our proposed method, \texttt{SynthBH} (\Cref{alg:spbh}), is a Benjamini--Hochberg type step-up procedure that uses \emph{rank-adaptive} guardrails. 
When considering $k$ candidate rejections, it only allows each p-value to be reduced by at most $k\eps/m$. 
This adaptive calibration is crucial: it allows the smallest p-values---those most likely to be non-nulls---to benefit from synthetic data, while preventing synthetic data from causing an uncontrolled proliferation of false discoveries.
Under a natural positive dependence condition (a mild extension of classical conditions), 
we prove that \texttt{SynthBH} controls the FDR at a level no larger than $\alpha+\eps$ in finite samples, \emph{without imposing any validity assumptions on the synthetic p-values themselves}.
When synthetic data are informative, \texttt{SynthBH} can yield substantially more discoveries than applying BH to the real data alone, with error $\alpha$;
when synthetic data are uninformative or misleading, the guardrails ensure the error remains controlled at the user-specified $\alpha+\eps$ level in the worst-case.

{\bf Why FDR control is different from prior results on synthetic-powered inference.}
Earlier work introduced synthetic-powered inference principles for general (monotone) loss control and hypothesis testing (the \gespi\ framework; \citealp{bashari2025statistical}) and for conformal prediction (synthetic-powered predictive inference; \citealp{bashari2025synthetic}).
A recurring theme there is that one can \emph{guardrail} the influence of synthetic data by intersecting/AND-ing it with a slightly more permissive real-data procedure at level $\alpha+\eps$, ensuring that the worst-case error is inflated by at most $\eps$.
In fact, this perspective inspired our definition of synthetic-powered p-values.
However, extending this idea to FDR is highly nontrivial:
FDR is a \emph{ratio} that depends on the full data-dependent rejection set and does not correspond to a monotone loss---i.e., enlarging the rejection set does not necessarily decrease the false discovery proportion. As a result, existing synthetic-powered inference theory does not directly apply.
Thus, new methodology and completely new analysis are needed to obtain FDR control 
while still leveraging synthetic evidence to improve power.

\subsection{Key Contributions}

This paper makes the following contributions.

{\bf A synthetic-powered BH procedure with finite-sample FDR control.}
We introduce \texttt{SynthBH} (\Cref{alg:spbh}), a multiple testing procedure that takes as input both real-data p-values $(p_j)$ and pooled real and synthetic p-values $(\tilde p_j)$.
The method uses the guarded synthetic-powered p-values \eqref{eqn:spp} with a rank-dependent guardrail $\delta=k\eps/m$ inside a BH-style step-up rule \eqref{eqn:rej_num}.

{\bf Distribution-free robustness to synthetic data quality.}
We prove that \texttt{SynthBH} controls the FDR in finite samples under an extended PRDS-type dependence condition (Definition~\ref{def:prds_ext}).
Crucially, our guarantee does \emph{not} require $\tilde p_j$ to be valid under the null; the synthetic data may be arbitrarily biased or shifted.
The role of $\eps$ is to quantify the permissible additional tolerance for incorporating synthetic evidence.

{\bf Weighted extension.}
Many applications assign heterogeneous importance to hypotheses (e.g., using prior biological knowledge, study design considerations, or multi-stage pipelines).
We therefore propose a weighted analogue (\Cref{alg:spbh_weighted}) and show an FDR bound that cleanly separates the base level $\alpha$ from the synthetic admission term $\eps$ (Theorem~\ref{thm:spbh}).

{\bf A concrete application with automatic verification of the dependence condition.}
We instantiate \texttt{SynthBH} for conformal outlier detection, where $p_j$ and $\tilde p_j$ are conformal p-values computed from a clean inlier reference set and from a larger merged set (possibly constructed from contaminated data via trimming).
In this setting, we verify the required positive dependence property and obtain distribution-free FDR control (Theorem~\ref{thm:outlier_prds}), extending conformal outlier detection to synthetic-powered regimes.

{\bf Empirical evaluation.}
We evaluate \texttt{SynthBH} on tabular outlier detection benchmarks and on genomic analyses of drug--cancer sensitivity associations, and complement these studies with controlled simulations.
Across these settings, \texttt{SynthBH} achieves empirical FDR close to the target level while improving power when synthetic data are informative, and it degrades gracefully when synthetic data quality deteriorates.
\footnote{Code to reproduce the experiments in this paper
is available at \href{https://github.com/Meshiba/synth-bh}{https://github.com/Meshiba/synth-bh}.}

\section{Related Work}

{\bf Multiple testing and FDR control.}
False discovery rate control was introduced by \citet{benjamini1995controlling}, and the Benjamini--Hochberg (BH) step-up procedure remains the most widely used approach due to its simplicity and strong guarantees.
Beyond independence, classical validity results rely on positive dependence structures such as PRDS \citep{benjamini2001control} and related formulations \citep{finner2009false}.
Our analysis builds on this tradition: we propose an extended PRDS condition (Definition~\ref{def:prds_ext}) tailored to settings with two collections of p-values $(p_j)$ and $(\tilde p_j)$, and we use it to establish finite-sample FDR control for \texttt{SynthBH}.

{\bf Using auxiliary information in multiple testing.}
There is a large literature on improving multiple testing power by incorporating side information, prior studies, or covariates---for example through p-value weighting and related optimization frameworks \citep{spjotvoll1972optimality,benjamini1997multiple,roeder2009genome,dobriban2015optimal,basu2018weighted}.
Such approaches can yield substantial gains when the auxiliary information is informative, but typically require assumptions that ensure validity of the resulting procedure (e.g., independence between covariates and null p-values, or explicit modeling assumptions linking current and prior studies).
Our setting differs in a key way: the auxiliary object $\tilde p_j$ is produced by pooling with synthetic data whose distribution can be arbitrary and unknown, so treating $\tilde p_j$ as a classical p-value (or as a covariate satisfying independence conditions) is generally unjustified.

{\bf Synthetic data for statistically valid inference.}
This paper is part of a broader effort to develop inference methods that can safely leverage synthetic or weakly trusted data.
Synthetic-powered predictive inference \citep{bashari2025synthetic} develops a score-transport mechanism for conformal prediction that yields finite-sample coverage guarantees while benefiting from synthetic data when score distributions align.
More recently, the \gespi\ framework \citep{bashari2025statistical} formalizes synthetic-powered inference as a general wrapper for controlling losses and error rates, highlighting the role of guardrails at level $\alpha+\eps$.
The present work complements 
these ideas by developing methods for 
multiple testing with FDR control.
Unlike one-shot hypothesis testing or loss control---where AND/OR-style guardrails can often be analyzed directly---FDR control couples all decisions through the random rejection threshold, making it essential to develop new procedures that account for this global dependence.

{\bf Conformal p-values and outlier detection with FDR control.}
Conformal prediction \citep{vovk1999machine,vovk2005algorithmic} provides distribution-free, finite-sample valid p-values under exchangeability.
\citet{bates2023testing} leverage conformal p-values together with BH to obtain distribution-free FDR control for detecting outliers among test points.
Our conformal outlier detection application 
is inspired by this work, but 
addresses a different bottleneck: conformal p-values computed from a small clean reference set can be conservative, whereas larger auxiliary datasets are often available but contaminated.
By combining a valid real-data conformal p-value with a pooled-data conformal p-value through \texttt{SynthBH}, we retain distribution-free FDR control while improving power when the auxiliary data are informative.

{\bf Why we focus on BH-style baselines in experiments.}
Because our goal is \emph{distribution-free} FDR control in the presence of potentially arbitrary synthetic data, many alternative approaches that rely on the pooled p-values being valid (or approximately valid) do not constitute meaningful competitors: their guarantees can fail precisely in the regimes we target.
Accordingly, our comparisons focus on canonical baselines that are always well-defined: BH on real-data p-values, BH at an inflated level $\alpha+\eps$ (illustrating the effect of relaxing the target), 
and BH on pooled-data p-values (illustrating the risks of naively trusting synthetic data).
\texttt{SynthBH} sits between these extremes, providing a principled, safe way to capture synthetic-data gains when available.

\section{Synthetic-Powered Multiple Hypothesis Testing}

\subsection{Notations}
We write $\R$ to denote the real line. For a positive integer $d$, let $[d]=\{1,2,\ldots,d\}$. For $\mathbf{a}=(a_1,\ldots,a_d)$ and $\mathbf{b}=(b_1,\ldots,b_d)$ in $\R^d$, we write $\mathbf{a}\preceq \mathbf{b}$ if and only if $a_j\le b_j$ for all $j\in[d]$. This defines the usual coordinatewise partial order on $\R^d$.

\subsection{Review of Benjamini-Hochberg Procedure}

The Benjamini-Hochberg procedure constructs the rejection set as follows. Given $m$ $p$-values $p_1, \ldots, p_m$ and a predefined level $\alpha \in (0,1)$, define
\begin{equation*}
    k^* = \max\left\{k \in [m] : \frac{m \cdot {p}_{(k)}}{k} \leq \alpha\right\},
\end{equation*}
where $p_{(k)}$ denotes the $k$-th smallest $p$-value. The rejection set is then $\{i\in [m] : p_i \leq p_{(k^*)}\}$.

This procedure provably controls the FDR at $(m_0/m) \alpha$, where $m_0$ is the number of true null hypotheses. The guarantee holds under independence of the $p$-values, or under the weaker assumption of positive regression dependence on the true nulls (PRDS), which we formally define below.

Recall that a set $A \subset \R^{m}$ is called \emph{ increasing} if it is upward closed under the partial order defined by the relation $\preceq$, i.e., if, for all $x\in A$, it follows that we have $y\in A$ for all $x\preceq y$.
 We then have the following notion of positive regression dependence, 
 which has been instrumental in establishing the validity of the Benjamini-Hochberg procedure under dependence \cite{benjamini1995controlling,benjamini2001control}. 

\begin{definition}[\citet{benjamini2001control}]\label{def:prds}
For a set $I \subset \{1, \ldots, m\}$ and a random vector $X = (X_1, X_2, \ldots, X_m)$, $X$ is \textit{positively regression dependent (PRDS) on $I$} if, for all $k \in I$, the conditional probability $x \mapsto \PPst{X \in A}{X_k = x}$ is nondecreasing on its domain as a function of $x \in \R$ for any increasing set $A \subset \R^{m}$.
\end{definition}

There is also a somewhat weaker form of PRDS, 
that we will build on in our analysis:

\begin{definition}\label{def:prds_weak}
In the setting of Definition \ref{def:prds}, 
$X$ is \emph{weakly positively regression dependent on $I$} if, for all $k\in I$ and for any increasing set $A\subset\R^m$, 
$x \mapsto \PPst{X \in A}{X_k \le x}$ 
is nondecreasing on its domain.
\end{definition}

\subsection{The \texttt{SynthBH} Procedure}

Now we introduce our procedure. For any $\delta \ge 0$, define $\tilde{p}_j^\delta$ as in~\eqref{eqn:spp}, 
and let $\tilde{p}_{(k)}^\delta$ denote the $k$-th smallest element among $\tilde{p}_1^\delta, \ldots, \tilde{p}_m^\delta$. Then for predetermined levels $\alpha \in (0,1)$ and $\eps \in (0,1)$, define
\begin{equation}\label{eqn:rej_num}
    k^* = \max\left\{k \in [m] : \frac{m \cdot \tilde{p}_{(k)}^{k\eps/m}}{k} \leq \alpha\right\}.
\end{equation}
We set $k^*=0$ if no index satisfies the condition.
Observe that we compare the $k$-th order statistics of the synthetic powered p-values, adjusted in an adaptive way by $k\eps/m$.
 For the smaller p-values, 
 it turns out to be reasonable to adjust them by smaller values. 
We reject the hypotheses corresponding to the smallest values of $(\tilde{p}^{k^*\eps/m}_{j})_{j\in[m]}$, i.e., the set $\{j\in[m]: \tilde p_j^{k^*\eps/m}\le \tilde p_{(k^*)}^{k^*\eps/m}\}$. This yields at least $k^*$ rejections and may yield more if ties occur at the cutoff.
See \Cref{alg:spbh} for the outline of the procedure.

\begin{algorithm}
\caption{\texttt{SynthBH}: Synthetic-powered multiple testing with FDR control}
\label{alg:spbh}
{\bf Input:} Real data $\Dn= \{Z_1,Z_2,\ldots,Z_n\}$. Synthetic data $\tDn = (\tilde{Z}_1, \tilde{Z}_2, \ldots, \tilde{Z}_N)$. 
P-value-generating algorithms $(\Alg_j)_{j \in [m]}$ such that $p_j=\Alg_j(\Dn)$ is a valid p-value under $H_j$.
Levels $\alpha,\eps \in (0,1)$.

{\bf Step 1:} Compute $p_j = \Alg_j(\Dn)$ and $\tilde{p}_j = \Alg_j(\Dn \cup \tDn)$, for $j\in [m]$.

{\bf Step 2:} Set $k^* = 0$ and $\mathcal{R} = \emptyset$. Then for $k = 1,2,\ldots,m$:

{\bf Step 2-1:} Compute $\tilde{p}_j^{k\eps/m} = p_j \wedge (\tilde{p}_j \vee (p_j - k\eps/m))$, for $j \in [m]$.

{\bf Step 2-2:} Find the $k$-th smallest element $\tilde{p}_{(k)}^{k\eps/m}$ among $(\tilde{p}_j^{k\eps/m})_{j \in [m]}$.

{\bf Step 2-3:} If $\tilde{p}_{(k)}^{k\eps/m} \leq \alpha k /m$, then update $k^* \leftarrow k$, $\mathcal{R} \leftarrow \{j \in [m] : \tilde{p}_j^{k\eps/m} \leq \tilde{p}_{(k)}^{k\eps/m}\}$.

{\bf Return:} Rejection set $\mathcal{R}$.
\end{algorithm}

Further, there are important settings where the p-values have predefined levels of importance, and we would like to apply different effective $\alpha$ levels for the different hypotheses. 
One way to capture this---leading to weighted hypothesis testing---is to assign weights $(w_j)_{j \in [m]}$ satisfying $w_j \geq 0$, $j \in [m]$,
to the different hypotheses;
 also satisfying the normalization conditions $\sum_{j=1}^m w_j = m$, 
 see e.g., \citet{spjotvoll1972optimality,benjamini1997multiple,roeder2009genome,dobriban2015optimal,basu2018weighted}, etc.
 
 Inspired by this line of work, we develop a method for FDR control
 with weighted hypotheses, enabling the use of synthetic data (\Cref{alg:spbh_weighted}).
 Here, the weights modulate the synthetic ``admission budget'' (how much each hypothesis may be boosted by synthetic data) rather than the base BH critical values; when $\eps=0$, the method reduces to the unweighted BH procedure at level $\alpha$.
 The intuition for this method is similar to the one for the previous case; with the difference that the 
 $j$-th p-value is shifted downwards by a term proportional to $w_j$, 
 hence effectively allowing us to compare it to a larger effective $\alpha$ when the weight $w_j$ is larger. 
 This procedure reduces to the previous one when all weights $w_j$ are equal to unity.
 However, our theoretical result for the general case is somewhat different from the one for the uniform case. 
 For clarity, we believe it is helpful to state the two methods separately.

\begin{algorithm}
\caption{Weighted \texttt{SynthBH}: Synth.-powered multiple testing with FDR control and heterogeneous weights}
\label{alg:spbh_weighted}
{\bf Input:} Real data $\Dn= \{Z_1,Z_2,\ldots,Z_n\}$. Synthetic data $\tDn = (\tilde{Z}_1, \tilde{Z}_2, \ldots, \tilde{Z}_N)$. Valid p-value-generating algorithms $(\Alg_j)_{j \in [m]}$, Levels $\alpha,\eps \in (0,1)$. Weights $(w_j)_{j \in [m]}$ satisfying $w_j \geq 0\, \forall j \in [m]$ and $\sum_{j=1}^m w_j = m$.

{\bf Step 1:} Compute $p_j = \Alg_j(\Dn)$ and $\tilde{p}_j = \Alg_j(\Dn \cup \tDn)$, for $j\in [m]$.

{\bf Step 2:} Set $k^* = 0$ and $\mathcal{R} = \emptyset$. Then for $k = 1,2,\ldots,m$:

{\bf Step 2-1:} Compute $\tilde{p}_j^{k w_j\eps/m} = p_j \wedge (\tilde{p}_j \vee (p_j - k w_j \eps/m))$, for $j \in [m]$.

{\bf Step 2-2:} Find the $k$-th smallest element $\tilde{p}_{(k)}^{k,\eps,w}$ among $(\tilde{p}_j^{k w_j\eps/m})_{j \in [m]}$.

{\bf Step 2-3:} If $\tilde{p}_{(k)}^{k,\eps,w} \leq \alpha k /m$, then update $k^* \leftarrow k$, $\mathcal{R} \leftarrow \{j \in [m] : \tilde{p}_j^{k w_j \eps/m} \leq \tilde{p}_{(k)}^{k,\eps,w}\}$.

{\bf Return:} Rejection set $\mathcal{R}$.
\end{algorithm}

\subsection{Computational Complexity}

We now turn to the computational aspects of \texttt{SynthBH}.
Although the rank-adaptive formulation may seem computationally intensive at first, \Cref{app-sec:complexity} shows that it can be exactly reduced to a single run of the classical BH procedure on a carefully constructed set of modified p-values (see \Cref{app-sec:complexity} for details).  
Thus, \texttt{SynthBH} can be implemented with the same $\mathcal{O}(m\log m)$ time and $\mathcal{O}(m)$ memory complexity as BH, while still leveraging synthetic data when beneficial and guaranteeing valid FDR control, as shown in the next section.

\subsection{Theoretical Guarantee}

We introduce the following condition, which is central to our theoretical guarantee.
This extends the PRDS condition from~\citet{finner2009false}, by allowing the conditioning to be done on a different vector $Y$ of random variables.

\begin{definition}\label{def:prds_ext}
Fix $m\in\mathbb{N}$ and let $I\subseteq[m]$.
Let $X\in\R^{d}$ and $Y=(Y_1,\ldots,Y_m)\in\R^m$ be random vectors.
We say that $X$ is \emph{positively regression dependent on $I$ with respect to $Y$} if, for all $k\in I$ and for any increasing set $A\subset\R^{d}$, the function
$x \mapsto \PPst{X\in A}{Y_k\le x}$
is nondecreasing on its domain.
\end{definition}

Equipped with these preliminaries, 
we now state our main result:
if the set of real and synthetic p-values are jointly PRDS with respect to the set of null real p-values, then our
algorithm controls the false discovery rate.

\begin{theorem}[\texttt{SynthBH} controls the FDR]\label{thm:spbh}
Suppose $(p_1,\ldots,p_m)$ are valid p-values for $H_1,\ldots,H_m$, respectively. If $(p_1,p_2,\ldots,p_m,\tilde{p}_1, \ldots, \tilde{p}_m)$ are positively regression dependent on the set of nulls $I_0 = \{j \in [m] : H_j \text{ is true}\}$ with respect to $(p_1,\ldots,p_m)$ in the sense of Definition~\ref{def:prds_ext}, then 
\texttt{SynthBH}
(\Cref{alg:spbh}) controls the false discovery rate at level $(m_0/m)\cdot(\alpha+\eps)$, where $m_0 = |I_0|$ denotes the number of true nulls. Furthermore, 
weighted \texttt{SynthBH} (\Cref{alg:spbh_weighted})
controls the false discovery rate at level $(m_0/m)\alpha + \frac{\eps}{m}\sum_{j\in I_0}w_j \le (m_0/m)\alpha+\eps$.
\end{theorem}

Such PRDS conditions are widely used in the literature on FDR control.
In the next section, we provide a concrete example in which this condition holds.
Moreover, \Cref{thm:spbh} shows that \texttt{SynthBH} controls the FDR at a level
adaptive to the true number of nulls, with slightly 
sharper control in the uniform case than in the weighted 
case.

\section{Application: Conformal Outlier Detection}\label{sec:app-cod}

Suppose we have data $\Dn = (X_1, \ldots, X_n) \iidsim P_X$. Given test points $(X_{n+1}, \ldots, X_{n+m})$, we aim to detect outliers by testing the following hypotheses:
\begin{equation}\label{eqn:null_outlier}
    H_j : X_{n+j} \sim P_X, \quad j = 1, 2, \ldots, m.
\end{equation}
\citet{bates2023testing} introduces a method based on conformal p-values that provides provable finite-sample, distribution-free FDR control.
The conformal p-value  for $H_j$ \cite{vovk1999machine,vovk2005algorithmic} is given by
\begin{equation}\label{eqn:conf_p}
    p_j = \frac{\sum_{i=1}^n \One{s(X_i) \geq s(X_{n+j})} + 1}{n + 1},
\end{equation}
where $s$ is a score function, which may be fixed or constructed using a separate split of the data. 
For example, $s(x)$ can be a score obtained by some outlier detection algorithm, such that a large value indicates that $x$ ``looks unusual'' for the machine learning model. 
The above $p_j$ is known to be super-uniform under $H_j$ \cite{vovk2005algorithmic}.
This p-value is always at least $1/(n+1)$, implying that when the sample size $n$ is small, the resulting detection procedure can be conservative.

Now, suppose we also have access to a large synthetic dataset
$\tDn = (\tilde{X}_1,\ldots,\tilde{X}_N)$. 
We can consider the following conformal p-value constructed from the merged data $\Dn \cup \tDn$:
\begin{equation}\label{eqn:conf_p_merged}
\tilde{p}_j = \frac{\sum_{i=1}^n \One{s(X_i) \geq s(X_{n+j})} + \sum_{i=1}^N \One{s(\tilde{X}_i) \geq s(X_{n+j})} + 1}{n + N + 1}.
\end{equation}

The variable $\tilde{p}_j$ can be expected to be less conservative, as its smallest possible value is $1/(n+N+1)$, but it is not guaranteed to be super-uniform under $H_j$.

\begin{algorithm}
\caption{Synthetic-powered conformal outlier detection}
\label{alg:sp_outlier}
{\bf Input:} Real data $\Dn= \{X_1,X_2,\ldots,X_n\}$. Synthetic data $\tDn = (\tilde{X}_1, \tilde{X}_2, \ldots, \tilde{X}_N)$. Score function $s : \X \rightarrow \R$, Levels $\alpha,\eps \in (0,1)$.

Run the \texttt{SynthBH} procedure with $p_j$ and $\tilde{p}_j$ computed according to~\eqref{eqn:conf_p} and~\eqref{eqn:conf_p_merged}, respectively, for $j=1,2,\ldots,m$; as well as all other given input.
{\bf Return:} Set of detected outliers $\mathcal{R}$.
\end{algorithm}

Nonetheless, 
we propose a Benjamini--Hochberg-type outlier detection method (Algorithm \ref{alg:sp_outlier}) that controls the false discovery rate. This guarantee follows from our general FDR control results, as we show that the associated conformal p-values satisfy the required PRDS property.

In order to accomplish this, we require 
the following lemma, showing that the classic PRDS condition (Definition~\ref{def:prds}), which conditions on exact values, implies the weaker conditioning-on-$\le$ form (Definition~\ref{def:prds_weak}).
 
\begin{lemma}\label{lem:ext}
    For a fixed measurable set $A \subset \R^n$, a random vector $X \in \R^n$ and a random variable $W$, if the function $x \mapsto \PPst{X \in A}{W = x}$ is nondecreasing, then $x \mapsto \PPst{X \in A}{W \leq x}$ is also nondecreasing.
\end{lemma}

Equipped with this, we can show the following result.

\begin{theorem}[Outlier detection via \texttt{SynthBH} with FDR control]\label{thm:outlier_prds}
Consider the hypotheses $(H_j)_{j \in [m]}$ defined in~\eqref{eqn:null_outlier}. Let $p_j$ and $\tilde{p}_j$ be as defined in~\eqref{eqn:conf_p} and~\eqref{eqn:conf_p_merged}, respectively, for $j \in [m]$. 
Suppose $(X_1,\ldots,X_n)$ are i.i.d.\ from $P_X$, the test points $X_{n+1},\ldots,X_{n+m}$ are mutually independent and independent of $(X_1,\ldots,X_n)$, 
and the synthetic 
data $(\tilde X_1,\ldots,\tilde X_N)$ is independent of $(X_1,\ldots,X_{n+m})$.
Let the score function $s$ be fixed 
and assume that the scores $(s(X_i))_{i\in[n+m]}$ and $(s(\tilde X_i))_{i\in[N]}$ are almost surely all distinct.
Then the random vector $(p_1, p_2, \ldots, p_m, \tilde{p}_1, \ldots, \tilde{p}_m)$ is positively regression dependent on the set of true nulls $I_0 = \{ j \in [m] : H_j \text{ is true} \}$ with respect to $(p_1, \ldots, p_m)$, in the sense of Definition~\ref{def:prds_ext}. Consequently, 
\Cref{alg:sp_outlier} controls the $\FDR$ at level $(m_0/m)\cdot(\alpha+\eps)$.

\end{theorem}

\begin{remark}
    The assumption that the scores are almost surely distinct can always be satisfied by adding a negligible random noise to the scores.
\end{remark}

\FloatBarrier

\section{Experiments}\label{sec:exp}

We demonstrate the performance of our method on conformal outlier detection (\Cref{sec:exp-od}) and on a genomic analysis of drug-cancer associations (\Cref{sec:exp-gdsc}). We defer results from controlled experiments on simulated data to \Cref{app-sec:exp-simulated}, where we further study the performance of \texttt{SynthBH} under varying settings.
We emphasize that the finite-sample worst-case guarantee for \texttt{SynthBH} is $\FDR \le \alpha+\eps$ (more precisely, $\FDR \le (m_0/m)(\alpha+\eps)$), while in many benign regimes we empirically observe FDR closer to the base level $\alpha$.

{\bf Methods} In each example, we compare the following methods:
\begin{compactitem}
\item \texttt{BH~(real)}: The classic BH procedure applied on p-values computed from the limited real data.
\item \texttt{BH~(real$+\eps$)}: Same as \texttt{BH~(real)}, but applies the BH procedure at a relaxed level $\alpha+\eps$.
\item \texttt{BH~(synth)}: The classic BH procedure applied on p-values computed from the pooled real and synthetic data.
This approach benefits from a larger sample size but does not provide any FDR control guarantees.
\item \texttt{SynthBH}: The proposed method, which leverages both real and synthetic data and enjoys FDR control guarantees.
\end{compactitem}

\subsection{Conformal Outlier Detection}\label{sec:exp-od}

In this section, we evaluate the performance of 
our method 
on the conformal outlier detection task described in \Cref{sec:app-cod}. 
Given real data $\Dn = (X_1, \ldots, X_n) \iidsim P_X$ and $m$ test points $(X_{n+1}, \dots, X_{n+m})$, the goal is to test, for each test point $X_{n+j}$, whether it is drawn from the inlier distribution $P_X$ or from a different, outlier distribution $Q_X\neq P_X$.

Conformal outlier detection provides finite-sample, distribution-free FDR control when the conformal p-values (as defined in~\eqref{eqn:conf_p}) are computed solely using inlier data. However, as discussed in \Cref{sec:app-cod}, this procedure can be overly conservative when the number of inliers, $n$, is small. In practice, 
obtaining a large collection of labeled inliers may be challenging.
In contrast, obtaining a large, potentially \textit{contaminated dataset} (i.e., containing a small fraction of outliers) is often readily available.
While such a set does not consist solely of inliers, 
we are able to leverage it as synthetic data.

As argued in \citet{bashari2025robust}, using a contaminated set directly often results in conservative inference due to the presence of outliers. To mitigate this, we adopt a cheap, annotation-free approach, where we trim the top $\rho\%$ samples (for $\rho$ defined below) 
of $\tDn$ that are 
marked by an outlier detection model, and treat the remaining samples as synthetic data. As illustrated in \citet{bashari2025robust}, this trimmed set cannot be treated as if it contains solely inliers, but provides a practical heuristic for leveraging large contaminated sets (which enjoys theoretical guarantees as per our results).

{\bf Datasets and experimental setup.}
We evaluate performance
on three tabular benchmark datasets for outlier detection: \textit{Shuttle} \citep{shuttle}, \textit{KDDCup99} \citep{KDDCup99}, and \textit{Credit-card} \citep{creditcard}. The target FDR level is $\alpha=10\%$, and $\eps=10\%$. All additional details are provided in \Cref{app-sec:datasets}.

\Cref{fig:cod} presents the performance of the conformal outlier detection methods on the three tabular datasets. All datasets show a similar trend. 
\texttt{BH~(synth)} leads to arbitrary---uncontrolled---FDR levels that depend on the unknown quality of the pooled synthetic and real data, highlighting the need for safe use of synthetic data. \texttt{BH~(real)} and \texttt{BH~(real$+\eps$)} suffer from the limited size of the real dataset: \texttt{BH~(real)} conservatively controls the FDR at level $\alpha$ and achieves lower power, while \texttt{BH~(real$+\eps$)} controls the FDR at a higher level but exhibits high variability.
In contrast, \texttt{SynthBH} achieves empirical FDR close to the nominal level $\alpha$ with higher power. When synthetic data are of high quality---e.g., as in the \textit{Shuttle} dataset---\texttt{SynthBH} achieves higher power than \texttt{BH~(real)} and lower variability in error than \texttt{BH~(real$+\eps$)}.

\begin{figure}[!h]
    \centering
    \includegraphics[width=0.36\linewidth, valign=t]{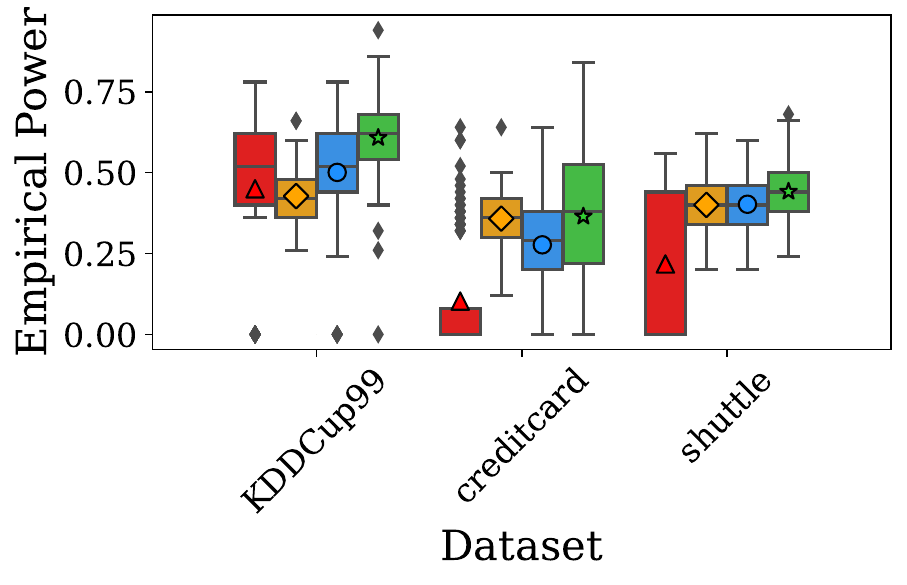}
    \includegraphics[width=0.36\linewidth, valign=t]{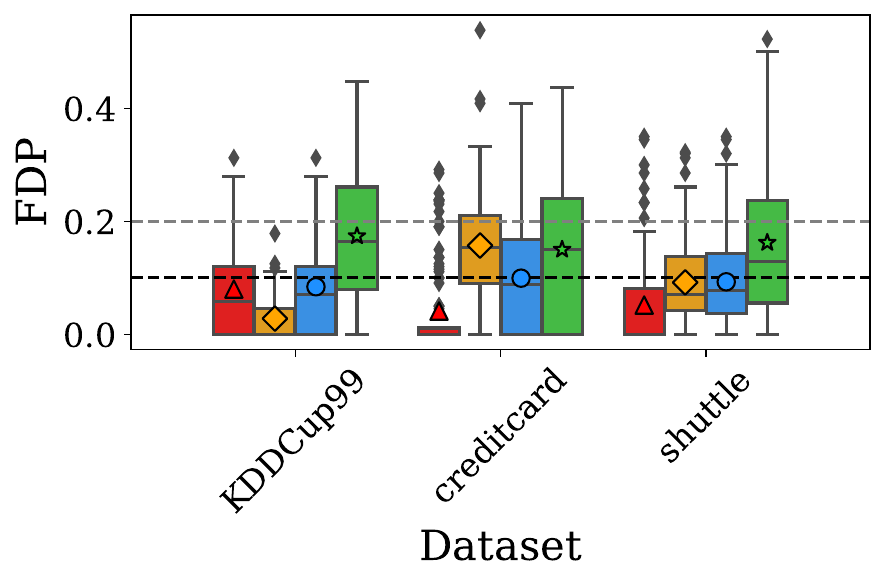}
    \includegraphics[width=0.15\linewidth, valign=t]{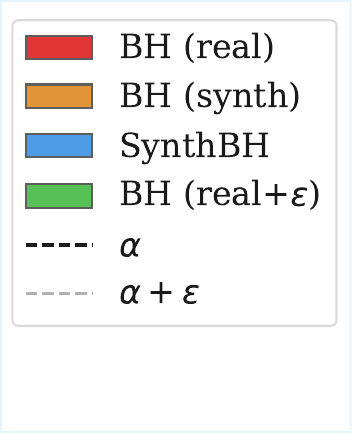}
    \caption{Comparison of conformal outlier detection methods on three tabular datasets. The target FDR level is $\alpha=10\%$ and $\eps=10\%$. Top: detection rate (empirical power). Bottom: false discovery proportion.}
    \label{fig:cod}
\end{figure}

We include additional results in \Cref{app-sec:od-exp}, showing performance as a function of the real data sample size and as a function of the trimming proportion, with the latter effectively illustrating performance under varying synthetic data quality.

To further illustrate that safely leveraging high-quality synthetic data can not only achieve FDR close to the nominal level but also improve power, \Cref{fig:od-fdr-power} presents empirical power as a function of the empirical FDR. 
\begin{figure*}[!h]
    \centering
    \begin{subfigure}[t]{0.28\textwidth}
    \includegraphics[width=\linewidth, valign=t]{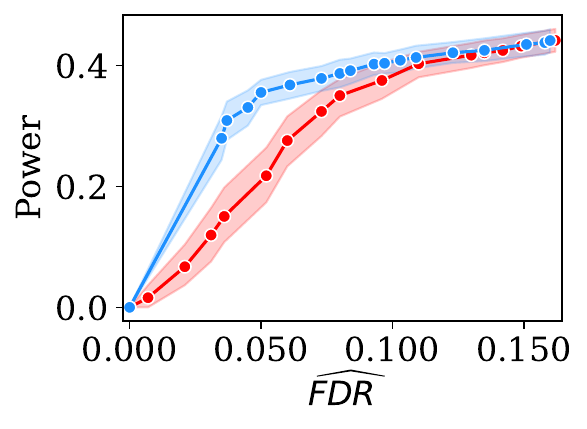}
    \includegraphics[width=\linewidth, valign=t]{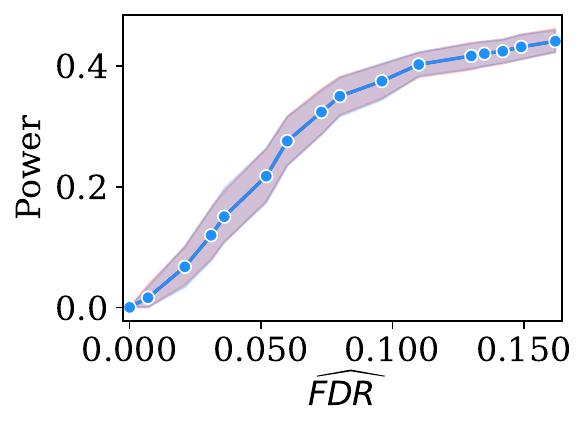}
    \caption{Shuttle}
    \end{subfigure}
    \begin{subfigure}[t]{0.28\textwidth}
    \includegraphics[width=\linewidth, valign=t]{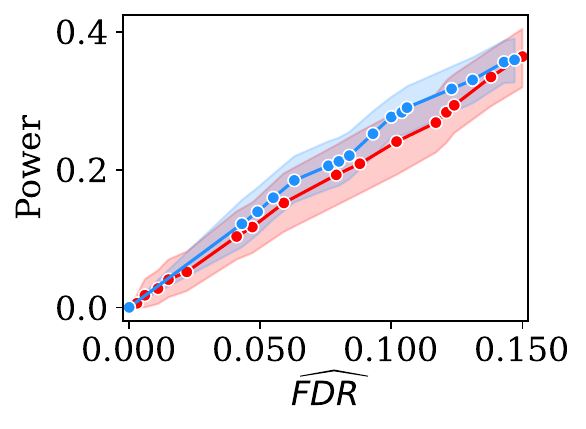}
    \includegraphics[width=\linewidth, valign=t]{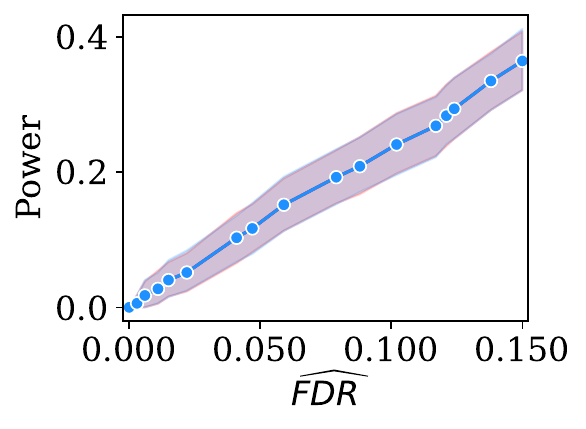}
    \caption{Credit-card}
    \end{subfigure}
    \begin{subfigure}[t]{0.28\textwidth}
    \includegraphics[width=\linewidth, valign=t]{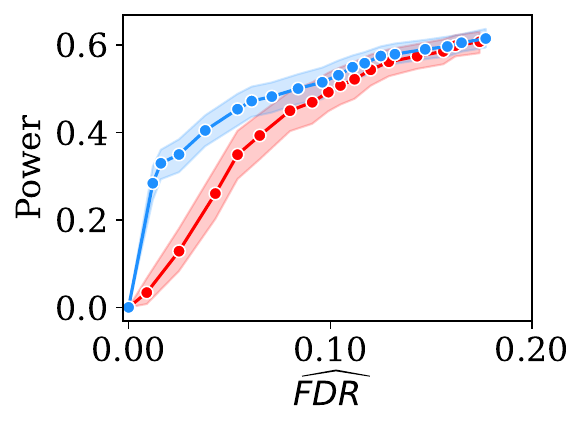}
    \includegraphics[width=\linewidth, valign=t]{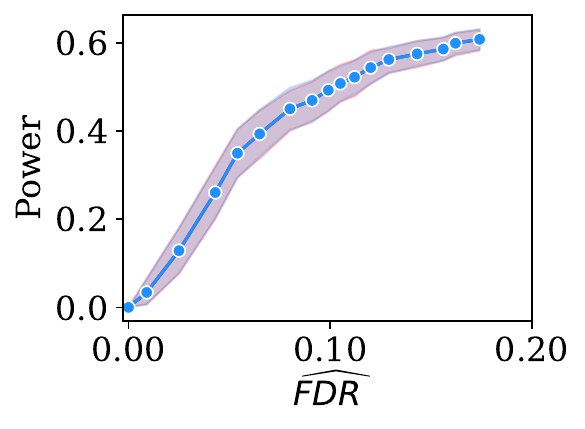}
    \caption{KDDCup99}
    \end{subfigure}
    \includegraphics[width=0.13\linewidth, valign=t]{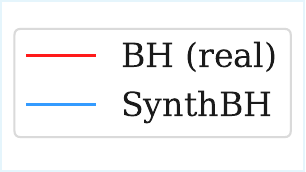}
    \caption{Power as a function of the empirical FDR for \texttt{BH~(real)} and our proposed method \texttt{SynthBH} across different outlier detection datasets: (a) Shuttle, (b) Credit-card, and (c) KDDCup99. The first row corresponds to a trimming proportion of 2\%, and the second row corresponds to a trimming proportion of 0\%, representing a scenario where synthetic data are not useful. Here, $\alpha\in [0,20\%]$.}
    \label{fig:od-fdr-power}
\end{figure*}
We note that this experiment is not feasible in practice, as it requires knowing the false discoveries, but it is included to compare our proposed method with using only the real data when both methods attain the same empirical FDR.

Following that figure, the first row shows performance when the trimming proportion is $\rho=2\%$, representing a scenario in which the synthetic data are informative. In this setting, \texttt{SynthBH} improves power relative to using only real data while also exhibiting lower variance. This can be explained as follows: by design, the effective p-values used by our method are always less than or equal to those computed from only the real data. When the synthetic data are of high quality, the larger sample size can lead to smaller p-values for non-null hypotheses. This can results in a larger rejection set with more true discoveries while maintaining the same false discovery proportion.

The second row of that figure presents 
a scenario in which the synthetic data are not informative, with trimming proportion set to $\rho = 0\%$. Here, the synthetic data provides no additional signal, and \texttt{SynthBH} matches the power of the method that utilize only the real data.
Additional results with a lower $\eps$ value are provided in \Cref{app-sec:od-exp}, showing the same trends.

\subsection{Genomics of Drug Sensitivity in Cancer}\label{sec:exp-gdsc}

We now demonstrate the performance of \texttt{SynthBH} on a genomics task: testing associations between genomic features and sensitivity to anticancer drugs. 
Identifying such associations helps characterize the molecular features that influence drug response in cancer cells and can ultimately enable the design of improved cancer therapies.

\textbf{Data and experimental setting}. We use data from the Genomics of Drug Sensitivity in Cancer (GDSC) \citep{yang2012genomics}, which contains drug response measurements (IC50 values) across a large collection of cancer cell lines, along with genomic annotations such as mutations.

\textbf{Hypotheses}. For each cancer type, genomic feature, and drug, we test whether drug response---measured by IC50 values across cell lines---is associated with the presence of the feature. Formally, the null hypothesis is
\begin{equation*}
    \mathcal{H}_{0}^{(c,f,d)}: \mathbb{E}_{f\text{ present}}{(\text{IC}50)} = \mathbb{E}_{f\text{ absent}}{(\text{IC}50)}.
\end{equation*}
To test this hypothesis, we partition the cell lines of cancer type $c$ into two groups according to the presence or absence of feature $f$, and apply a two-sample $t$-test to obtain a $p$-value for each $(c,f,d)$ triple. 
This follows the general scheme of the GDSC~\citep{yang2012genomics} paper and the  GDSCTools python package~\citep {cokelaer2018gdsctools}; since we consider tissue-specific hypotheses, we use a two-sample $t$-test instead of ANOVA---used in these works.
We consider 100 hypotheses of interest for breast invasive carcinoma and lung adenocarcinoma, listed in~\Cref{app-sec:gdsc-details}.

\textbf{Pan-cancer as synthetic data}. 
The number of cell lines per cancer type can be limited, while pan-cancer data across multiple tissues provides abundant auxiliary information.
Pan-cancer associations are informative when a genomic feature influences drug response through shared biological mechanisms across cancers, but uninformative or misleading when drug sensitivity is cancer-type specific. We use data from all other tissues as synthetic data.

Since ground-truth labels for each hypothesis are unavailable and the number of real cell lines is limited, we adopt the following scheme to compute an empirical proxy. 
Specifically, for each hypothesis, we compute a \emph{ground-truth score}, defined as the proportion of runs in which the hypothesis is rejected by the BH procedure at level~$\alpha$. A higher ground-truth score indicates stronger evidence that the hypothesis is non-null. To estimate this score, we repeat the following procedure 10 times: randomly sample $80\%$ of the cell lines for each target cancer tissue, and test the 100 hypotheses of interest.

For comparing the different methods, we randomly sample 50\% of both real (cancer-type-specific) and synthetic (pan-cancer) cell lines to compute the real and pooled p-values, and then apply each method with $\alpha=10\%$ and $\eps=10\%$. We report the number of rejections and the average ground-truth score of the rejection set, computed over 10 trials.

\Cref{fig:gdsc} shows the performance of multiple testing methods on the GDSC genomic data. On the left, \texttt{SynthBH} yields more rejections than \texttt{BH~(real)}. On the right, the average ground-truth score of the rejection sets obtained by \texttt{SynthBH} is high, and notably larger than that of \texttt{BH~(real$+\eps$)}. This indicates that the additional rejections of \texttt{SynthBH} likely correspond to meaningful true discoveries rather than false positives. Interestingly, \texttt{BH~(synth)} alone produces fewer rejections with lower ground-truth scores, though it still provides useful information for some hypotheses, as evidenced by the performance of \texttt{SynthBH}. We provide additional results with $\eps=5\%$ in \Cref{app-sec:gdsc-exp}, which present a trend similar to that reported here.

\begin{figure}[!h]
    \centering
    \includegraphics[height=0.19\linewidth]{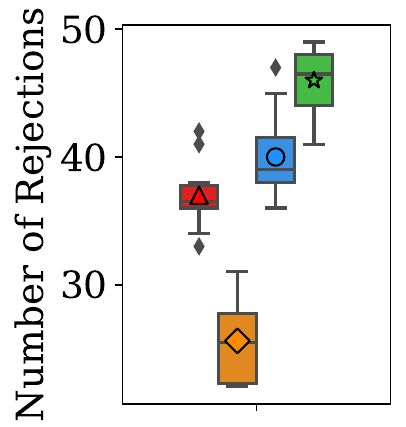}
    \includegraphics[height=0.19\linewidth]{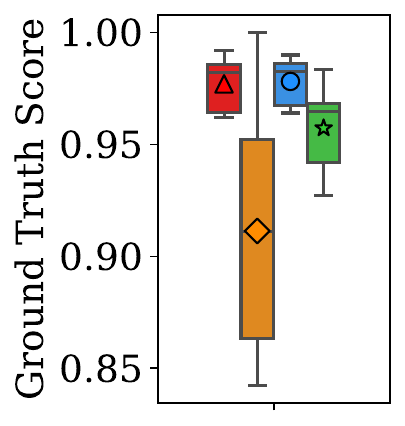}
    \includegraphics[height=0.19\linewidth]{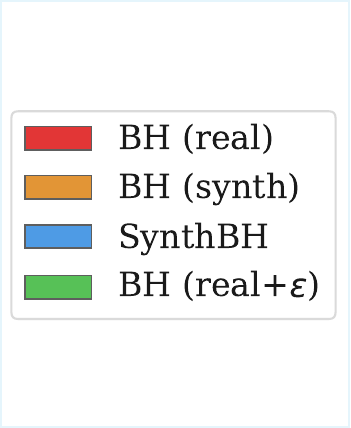}
    \caption{Comparison of multiple testing methods on GDSC genomic data for 100 tissue-feature-drug hypotheses. The target FDR level is $\alpha=10\%$ and $\eps=10\%$. Left: number of rejections. Right: average ground-truth score of the rejection set.}
    \label{fig:gdsc}
\end{figure}

\section{Discussion}

This work introduces \texttt{SynthBH}, a synthetic-powered multiple testing procedure that safely leverages auxiliary/synthetic data while guaranteeing finite-sample FDR control under a mild PRDS-type dependence condition, without requiring the pooled p-values to be valid.
Important future directions include deriving guarantees under weaker (or arbitrary) dependence, developing broadly verifiable sufficient conditions beyond the conformal outlier setting studied here, and designing principled ways to choose $\eps$ (or adapt it) and to integrate data-dependent synthetic filtering/weighting while preserving finite-sample validity.

\section*{Acknowledgments}
E.~D. and Y.~L. were partially supported by the US NSF, NIH, ARO, AFOSR, ONR, and the Sloan Foundation.
M.~B. and Y.~R. were supported by the European Union (ERC, SafetyBounds, 101163414). Views and opinions expressed are however those of the authors only and do not necessarily reflect those of the European Union or the European Research Council Executive Agency. Neither the European Union nor the granting authority can be held responsible for them. This research was also partially supported by the Israel Science Foundation (ISF grant 729/21). 
Y.~R. acknowledges additional support from the Career Advancement Fellowship at the Technion.

\FloatBarrier


\begin{thebibliography}{23}
\providecommand{\natexlab}[1]{#1}
\providecommand{\url}[1]{\texttt{#1}}
\expandafter\ifx\csname urlstyle\endcsname\relax
  \providecommand{\doi}[1]{doi: #1}\else
  \providecommand{\doi}{doi: \begingroup \urlstyle{rm}\Url}\fi

\bibitem[Bashari et~al.(2025{\natexlab{a}})Bashari, Lee, Lotan, Dobriban, and Romano]{bashari2025statistical}
M.~Bashari, Y.~Lee, R.~M. Lotan, E.~Dobriban, and Y.~Romano.
\newblock Statistical inference leveraging synthetic data with distribution-free guarantees.
\newblock \emph{arXiv preprint arXiv:2509.20345}, 2025{\natexlab{a}}.

\bibitem[Bashari et~al.(2025{\natexlab{b}})Bashari, Lotan, Lee, Dobriban, and Romano]{bashari2025synthetic}
M.~Bashari, R.~M. Lotan, Y.~Lee, E.~Dobriban, and Y.~Romano.
\newblock Synthetic-powered predictive inference.
\newblock \emph{Advances in Neural Information Processing Systems}, 2025{\natexlab{b}}.

\bibitem[Bashari et~al.(2025{\natexlab{c}})Bashari, Sesia, and Romano]{bashari2025robust}
M.~Bashari, M.~Sesia, and Y.~Romano.
\newblock Robust conformal outlier detection under contaminated reference data.
\newblock In \emph{Forty-second International Conference on Machine Learning}, 2025{\natexlab{c}}.

\bibitem[Basu et~al.(2018)Basu, Cai, Das, and Sun]{basu2018weighted}
P.~Basu, T.~T. Cai, K.~Das, and W.~Sun.
\newblock Weighted false discovery rate control in large-scale multiple testing.
\newblock \emph{Journal of the American Statistical Association}, 2018.

\bibitem[Bates et~al.(2023)Bates, Cand{\`e}s, Lei, Romano, and Sesia]{bates2023testing}
S.~Bates, E.~Cand{\`e}s, L.~Lei, Y.~Romano, and M.~Sesia.
\newblock Testing for outliers with conformal p-values.
\newblock \emph{The Annals of Statistics}, 2023.

\bibitem[Benjamini and Hochberg(1995)]{benjamini1995controlling}
Y.~Benjamini and Y.~Hochberg.
\newblock Controlling the false discovery rate: a practical and powerful approach to multiple testing.
\newblock \emph{Journal of the Royal statistical society: series B (Methodological)}, 1995.

\bibitem[Benjamini and Hochberg(1997)]{benjamini1997multiple}
Y.~Benjamini and Y.~Hochberg.
\newblock Multiple hypotheses testing with weights.
\newblock \emph{Scandinavian Journal of Statistics}, 1997.

\bibitem[Benjamini and Yekutieli(2001)]{benjamini2001control}
Y.~Benjamini and D.~Yekutieli.
\newblock The control of the false discovery rate in multiple testing under dependency.
\newblock \emph{The Annals of Statistics}, 2001.

\bibitem[Bischl et~al.(2021)Bischl, Casalicchio, Feurer, Gijsbers, Hutter, Lang, Mantovani, van Rijn, and Vanschoren]{OpenML2021}
B.~Bischl, G.~Casalicchio, M.~Feurer, P.~Gijsbers, F.~Hutter, M.~Lang, R.~G. Mantovani, J.~N. van Rijn, and J.~Vanschoren.
\newblock Open{ML}: A benchmarking layer on top of {OpenML} to quickly create, download, and share systematic benchmarks.
\newblock \emph{Advances in Neural Information Processing Systems}, 2021.

\bibitem[Catlett(1992)]{shuttle}
J.~Catlett.
\newblock {Statlog (Shuttle)}.
\newblock UCI Machine Learning Repository, 1992.

\bibitem[Cokelaer et~al.(2018)Cokelaer, Chen, Iorio, Menden, Lightfoot, Saez-Rodriguez, and Garnett]{cokelaer2018gdsctools}
T.~Cokelaer, E.~Chen, F.~Iorio, M.~P. Menden, H.~Lightfoot, J.~Saez-Rodriguez, and M.~J. Garnett.
\newblock {GDSCTools} for mining pharmacogenomic interactions in cancer.
\newblock \emph{Bioinformatics}, 2018.

\bibitem[Dobriban et~al.(2015)Dobriban, Fortney, Kim, and Owen]{dobriban2015optimal}
E.~Dobriban, K.~Fortney, S.~K. Kim, and A.~B. Owen.
\newblock Optimal multiple testing under a {G}aussian prior on the effect sizes.
\newblock \emph{Biometrika}, 2015.

\bibitem[Feurer et~al.(2021)Feurer, van Rijn, Kadra, Gijsbers, Mallik, Ravi, Mueller, Vanschoren, and Hutter]{OpenML2021p}
M.~Feurer, J.~N. van Rijn, A.~Kadra, P.~Gijsbers, N.~Mallik, S.~Ravi, A.~Mueller, J.~Vanschoren, and F.~Hutter.
\newblock Open{ML}-{P}ython: an extensible python api for openml.
\newblock \emph{JMLR}, 2021.

\bibitem[Finner et~al.(2009)Finner, Dickhaus, and Roters]{finner2009false}
H.~Finner, T.~Dickhaus, and M.~Roters.
\newblock On the false discovery rate and an asymptotically optimal rejection curve.
\newblock \emph{The Annals of Statistics}, 2009.

\bibitem[Group(2013)]{creditcard}
M.~L. Group.
\newblock {Credit Card Fraud Detection Data Set}, 2013.

\bibitem[Liu et~al.(2008)Liu, Ting, and Zhou]{liu2008isolation}
F.~T. Liu, K.~M. Ting, and Z.-H. Zhou.
\newblock Isolation forest.
\newblock In \emph{International Conference on Data Mining}. IEEE, 2008.

\bibitem[Roeder and Wasserman(2009)]{roeder2009genome}
K.~Roeder and L.~Wasserman.
\newblock Genome-wide significance levels and weighted hypothesis testing.
\newblock \emph{Statistical science: a review journal of the Institute of Mathematical Statistics}, 2009.

\bibitem[Spj{\o}tvoll(1972)]{spjotvoll1972optimality}
E.~Spj{\o}tvoll.
\newblock On the optimality of some multiple comparison procedures.
\newblock \emph{The Annals of Mathematical Statistics}, 1972.

\bibitem[Stolfo et~al.(1999)Stolfo, Fan, Lee, Prodromidis, and Chan]{KDDCup99}
S.~Stolfo, W.~Fan, W.~Lee, A.~Prodromidis, and P.~Chan.
\newblock {KDD Cup 1999 Data}.
\newblock UCI Machine Learning Repository, 1999.

\bibitem[Vovk et~al.(1999)Vovk, Gammerman, and Saunders]{vovk1999machine}
V.~Vovk, A.~Gammerman, and C.~Saunders.
\newblock Machine-learning applications of algorithmic randomness.
\newblock In \emph{{I}nternational {C}onference on {M}achine {L}earning}, 1999.

\bibitem[Vovk et~al.(2005)Vovk, Gammerman, and Shafer]{vovk2005algorithmic}
V.~Vovk, A.~Gammerman, and G.~Shafer.
\newblock \emph{Algorithmic learning in a random world}.
\newblock Springer Science \& Business Media, 2005.

\bibitem[Wang(2022)]{wang2022elementary}
R.~Wang.
\newblock Elementary proofs of several results on false discovery rate.
\newblock \emph{arXiv preprint arXiv:2201.09350}, 2022.

\bibitem[Yang et~al.(2012)Yang, Soares, Greninger, Edelman, Lightfoot, Forbes, Bindal, Beare, Smith, Thompson, et~al.]{yang2012genomics}
W.~Yang, J.~Soares, P.~Greninger, E.~J. Edelman, H.~Lightfoot, S.~Forbes, N.~Bindal, D.~Beare, J.~A. Smith, I.~R. Thompson, et~al.
\newblock Genomics of drug sensitivity in cancer ({GDSC}): a resource for therapeutic biomarker discovery in cancer cells.
\newblock \emph{Nucleic Acids Research}, 2012.

\end{thebibliography}

\clearpage

\appendix
\crefalias{subsection}{appendix}
\Crefname{appendix}{Appendix}{Appendices}
\crefname{appendix}{appendix}{appendices}
\crefalias{section}{appendix}
\Crefname{section}{Appendix}{Appendices}
\crefname{section}{appendix}{appendices}
\onecolumn

\renewcommand{\thefigure}{S\arabic{figure}}
\setcounter{figure}{0}

\renewcommand{\thetable}{S\arabic{table}}
\setcounter{table}{0}

\section{Additional Experiments}

\subsection{Experiments with Simulated Data}\label{app-sec:exp-simulated}

In this section, we present controlled experiments based on simulated data to further study the performance of \texttt{SynthBH} under varying settings. Specifically, we consider the following multiple hypothesis testing problem.

Consider $m$ null hypotheses $(\mathcal{H}_{0,i})_{i\in [m]}$ tested simultaneously against corresponding alternatives $(\mathcal{H}_{1,i})_{i\in [m]}$, where
\[
\mathcal{H}_{0,i}: q_i \leq 0.5 \text{ vs. } \mathcal{H}_{1,i}: q_i > 0.5.
\]

For each hypothesis $i$, we observe $n$ real samples drawn independently from $\text{Bernoulli}(q_i)$ and $N$ synthetic samples drawn independently from $\text{Bernoulli}(\tilde{q}_i)$. 
We then use a randomized binomial test to compute p-values, which are subsequently used by all methods.

Unless stated otherwise, 
we set $n=200$, $N=1000$, $\alpha=10\%$, and $\eps=10\%$. We test $m=1000$ hypotheses, of which 5\% follow the alternative ($q_i = 0.6$), and the rest follow the null. 
The synthetic samples for each hypothesis are generated from $\text{Bernoulli}(0.5)$ when $\mathcal{H}_{0,i}$ holds, and from $\text{Bernoulli}(\rho_{\text{synth}})$ with $\rho_{\text{synth}}=0.55$ otherwise.
This setup reflects a scenario in which the synthetic signal for non-null hypotheses is weaker than that of the real data, but may still be informative due to the larger synthetic sample size. We report the average detection rate (power) and the average false discovery proportion, computed over 100 independent trials.

\textbf{The effect of the real data sample size $n$}. \Cref{app-fig:simulated-n} presents the performance of different multiple hypothesis testing methods as a function of the real data sample size $n$. When 
$n$ is small (e.g., $n=50$), both \texttt{BH~(real)} and \texttt{BH~(real$+\eps$)} attain nearly zero power; as a consequence, \texttt{SynthBH} also attains zero power. As $n$ increases, the power of both \texttt{BH~(real)} and \texttt{BH~(real$+\eps$)} increases, with the former controlling the FDR at level $\alpha$ and the latter at the inflated level $\alpha+\eps$, and thus achieving higher power, as expected.

In this regime, the synthetic data are informative, and \texttt{SynthBH} attains power comparable to \texttt{BH~(real$+\eps$)}.
In this benign regime we empirically observe FDR close to the nominal level $\alpha$.
For sufficiently large $n$, all methods achieve power close to one.

\begin{figure}[!h]
    \centering
    \includegraphics[width=0.35\linewidth, valign=t]{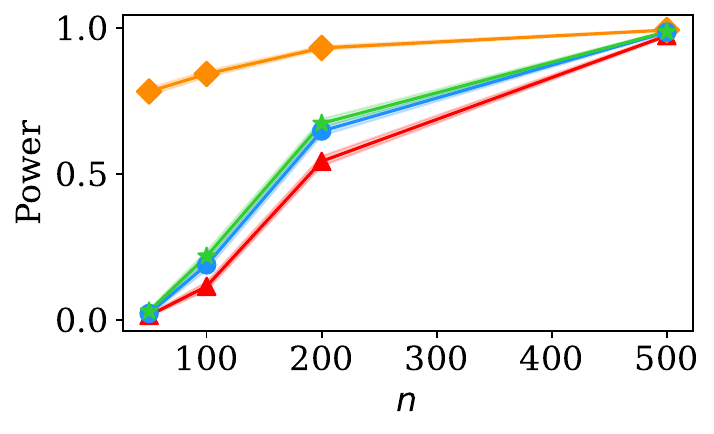}
    \includegraphics[width=0.35\linewidth, valign=t]{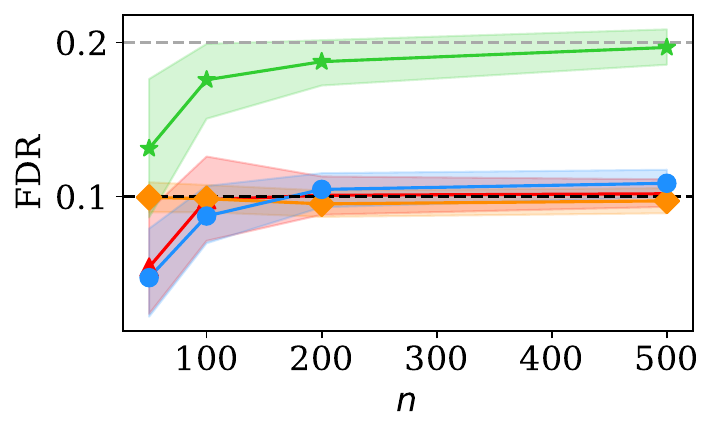}
    \includegraphics[width=0.15\textwidth, valign=t]{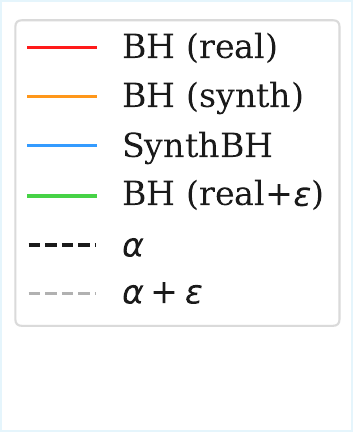}
    \caption{Comparison of multiple hypothesis testing as a function of the real data sample size $n$. Left:
average detection rate (power). Right: average false discovery proportion.}
    \label{app-fig:simulated-n}
\end{figure}
  
\textbf{The effect of the synthetic data quality}.
\Cref{app-fig:simulated-rho-synt} presents the performance as a function of the signal strength of the synthetic data. \Cref{app-fig:simulated-rho-synt-non-null} corresponds to a setting in which the synthetic data are informative for both null and non-null hypotheses, while \Cref{app-fig:simulated-rho-synt-non-and-null} illustrates an extreme scenario reflecting worst-case synthetic data (in terms of error), where synthetic data exhibit signal for all hypotheses, nulls and non-nulls; the magnitude of this signal is given by the x-axis.

We observe similar trends across both settings. 
In these simulations, \texttt{SynthBH} attains power at least that of \texttt{BH~(real)} and increases with the synthetic signal strength, with performance approaching that of \texttt{BH~(real$+\eps$)} as the synthetic signal strengthens.

In terms of FDR, in \Cref{app-fig:simulated-rho-synt-non-null}, where the synthetic data are of high quality, \texttt{SynthBH} achieves FDR close to the target level $\alpha$. In contrast, in \Cref{app-fig:simulated-rho-synt-non-and-null}, which reflects a worst-case error scenario for synthetic data, \texttt{BH~(synth)} exhibits extremely inflated FDR, while \texttt{SynthBH} controls the FDR at level $\alpha+\eps$, as guaranteed in \Cref{thm:spbh}.

\begin{figure}[!h]
    \centering
    \begin{subfigure}[t]{\textwidth}
    \hspace{3em}
    \includegraphics[width=0.35\linewidth, valign=t]{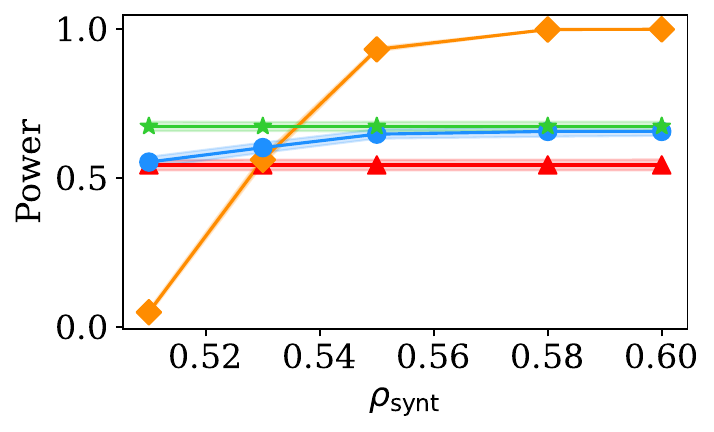}
    \includegraphics[width=0.35\linewidth, valign=t]{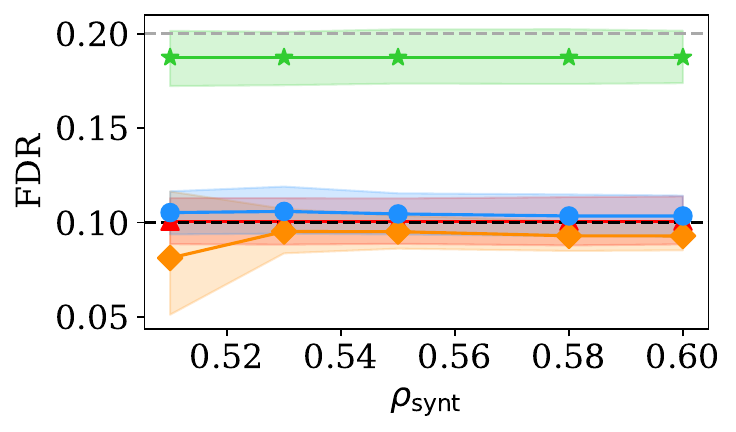}
    \caption{High-quality synthetic data: synthetic data follow the same null/non-null structure as the real data.}\label{app-fig:simulated-rho-synt-non-null}
    \end{subfigure}
    \begin{subfigure}[t]{\textwidth}
    \centering
    \includegraphics[width=0.35\linewidth, valign=t]{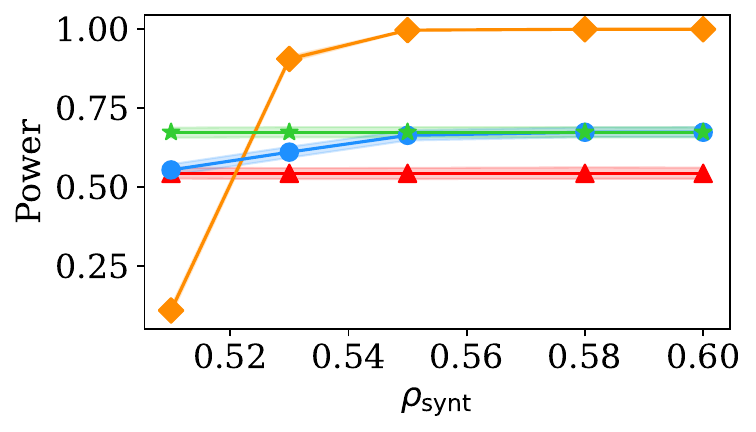}
    \includegraphics[width=0.35\linewidth, valign=t]{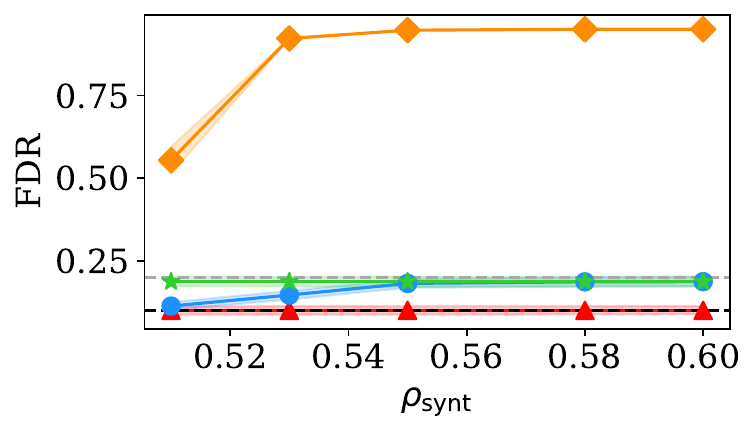}
    \includegraphics[width=0.15\textwidth, valign=t]{figures/experiments/simulated_data/legend.pdf}
    \caption{Worst-case synthetic data (in terms of error): synthetic data indicate non-null signal for all hypotheses, including nulls.}\label{app-fig:simulated-rho-synt-non-and-null}
    \end{subfigure}
    \caption{Comparison of multiple hypothesis testing as a function of the synthetic signal strength $\rho_{\text{synth}}$. Results are shown for (a) high-quality synthetic data and (b) worst-case synthetic data (in terms of error). Left: average detection rate (power). Right: average false discovery proportion.}
    \label{app-fig:simulated-rho-synt}
\end{figure}

\begin{figure}[!h]
    \centering
    \begin{subfigure}[t]{\textwidth}
    \hspace{3em}
    \includegraphics[width=0.35\linewidth, valign=t]{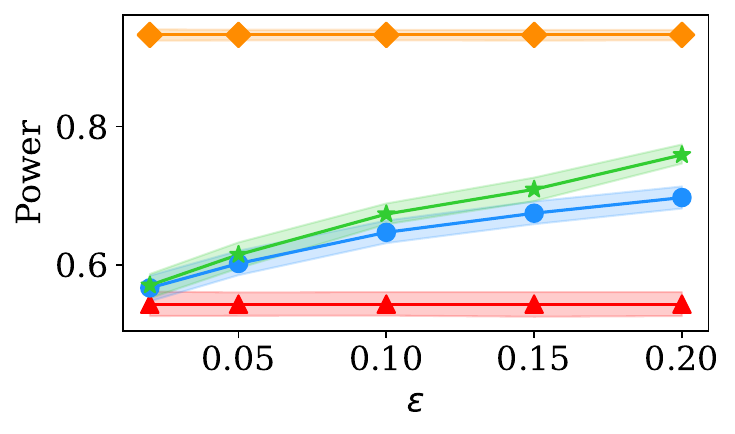}
    \includegraphics[width=0.35\linewidth, valign=t]{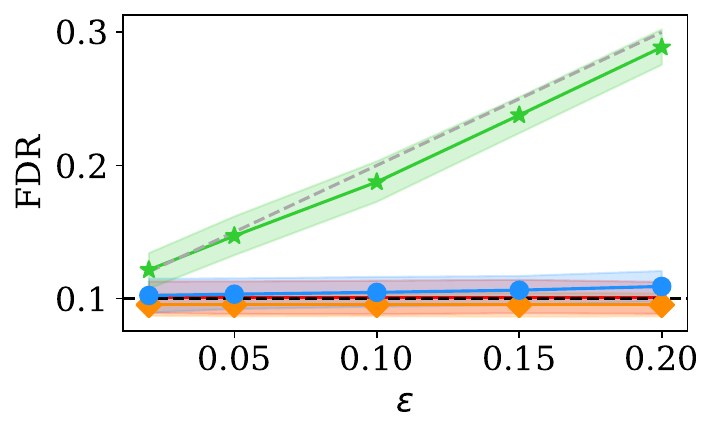}
    \caption{High-quality synthetic data: synthetic data follow the same null/non-null structure as the real data.}\label{app-fig:simulated-eps-non-null}
    \end{subfigure}
    \begin{subfigure}[t]{\textwidth}
    \centering
    \includegraphics[width=0.35\linewidth, valign=t]{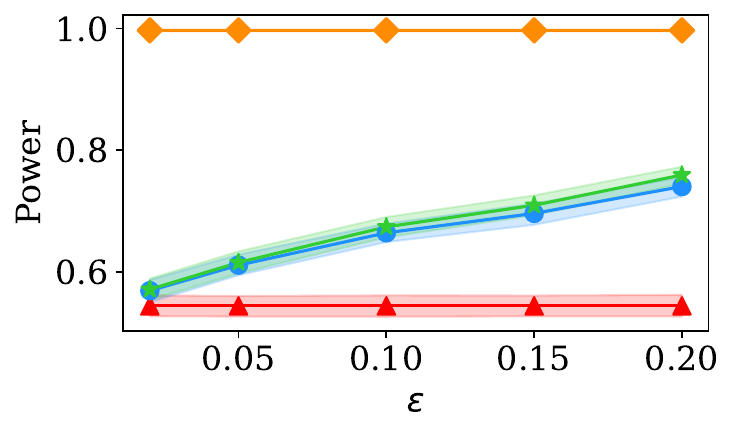}
    \includegraphics[width=0.35\linewidth, valign=t]{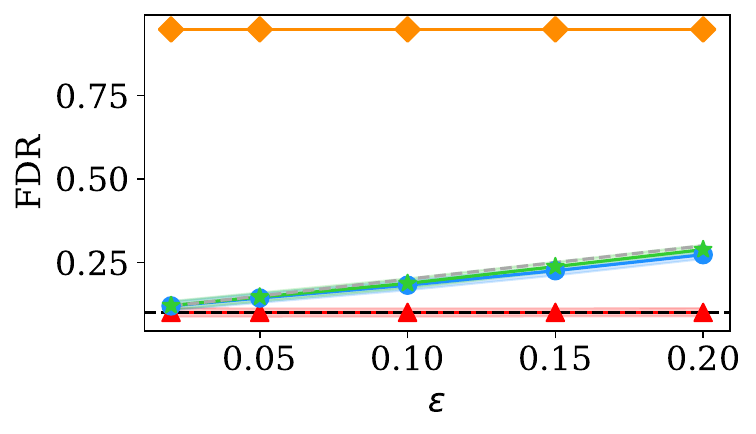}
    \includegraphics[width=0.15\textwidth, valign=t]{figures/experiments/simulated_data/legend.pdf}
    \caption{Worst-case synthetic data (in terms of error): synthetic data indicate non-null signal for all hypotheses, including nulls.}\label{app-fig:simulated-eps-non-and-null}
    \end{subfigure}
    \caption{Comparison of multiple hypothesis testing as a function of $\eps$. Results are shown for (a) high-quality synthetic data and (b) worst-case synthetic data (in terms of error). Left: average detection rate (power). Right: average false discovery proportion.}
    \label{app-fig:simulated-eps}
\end{figure}

\textbf{The effect of $\eps$}. As discussed throughout the paper, $\eps$ can be interpreted as the ``admission cost'' in terms of extra error the user is willing to accept to potentially benefit from synthetic data.
Therefore, in the worst-case scenario, \texttt{SynthBH} incurs an additional $\eps$ error, but at the same time its power increases with $\eps$ for all synthetic data qualities.

Similar to the previous analyses, we consider two scenarios: (a) high-quality synthetic data, and (b) worst-case synthetic data in terms of error. \Cref{app-fig:simulated-eps} shows performance as a function of the parameter $\eps$ for both scenarios. 
In \Cref{app-fig:simulated-eps-non-null}, \texttt{SynthBH} improves power as $\eps$ increases while maintaining FDR close to the nominal level $\alpha$ for high-quality synthetic data. In contrast, \Cref{app-fig:simulated-eps-non-and-null} illustrates the worst-case synthetic data scenario: the power of \texttt{SynthBH} also increases with $\eps$, but FDR is close to $\alpha+\eps$, remaining controlled. \texttt{BH~(synth)}, by comparison, exhibits extremely inflated and uncontrolled FDR.

\FloatBarrier

\subsection{Conformal Outlier Detection}\label{app-sec:od-exp}

We include additional results extending the experiments presented in \Cref{sec:exp-od}.

\Cref{app-fig:od-n} shows the performance of the conformal outlier detection methods on the \textit{Shuttle} dataset as a function of the real data sample size $n$. When $n$ is very small, both \texttt{BH~(real)} and \texttt{BH~(real$+\eps$)} achieve nearly zero power, and consequently \texttt{SynthBH} also has zero power. As $n$ increases, the power of all three methods improves. In this regime, \texttt{SynthBH} achieves higher power than \texttt{BH~(real)}, while maintaining FDR close to the nominal level $\alpha$, and exhibits lower variability than \texttt{BH~(real$+\eps$)}.

\begin{figure}[!h]
    \centering
    \includegraphics[width=0.35\linewidth, valign=c]{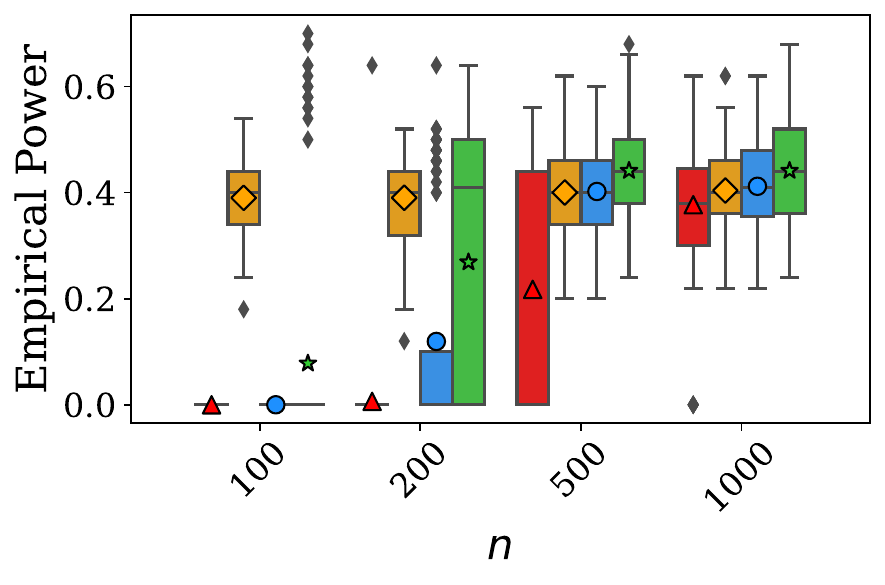}
    \includegraphics[width=0.35\linewidth, valign=c]{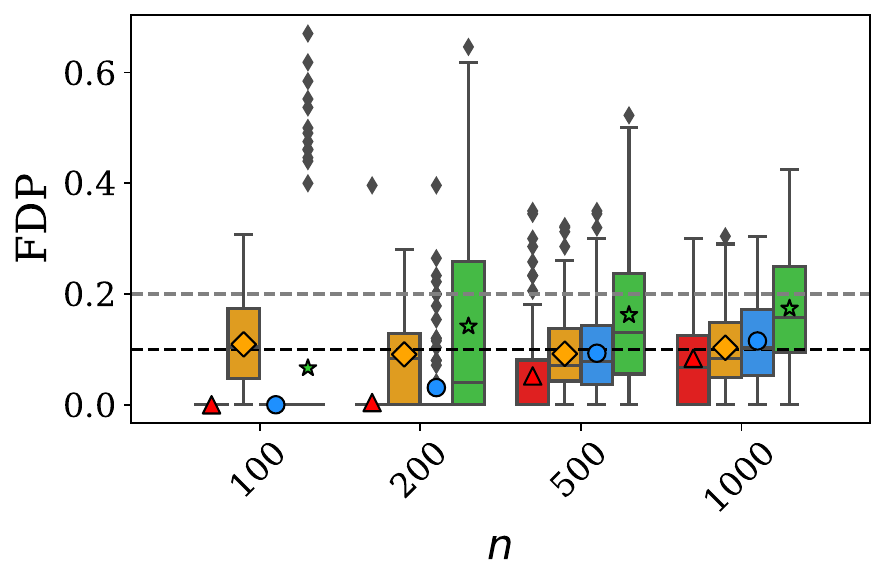}
    \includegraphics[width=0.15\linewidth, valign=c]{figures/experiments/new_outlier_detection/legend_box.pdf}
    \caption{Comparison of conformal outlier detection methods on the \textit{Shuttle} dataset as a function of the real data sample size $n$. Left: detection rate (empirical power). Right: false discovery proportion.}
    \label{app-fig:od-n}
\end{figure}

\Cref{app-fig:od-p_trim} shows the performance of the conformal outlier detection methods as a function of the trimming proportion $\rho$, which effectively controls the quality of the synthetic data. For low trimming values ($\rho = 1\%$), \texttt{BH~(synth)} conservatively controls the FDR at level $\alpha$. In this regime, \texttt{SynthBH} also remains conservative, yet notably achieves higher power than both \texttt{BH~(real)} and \texttt{BH~(synth)}.
\begin{figure}[!h]
    \centering
    \includegraphics[width=0.35\linewidth, valign=c]{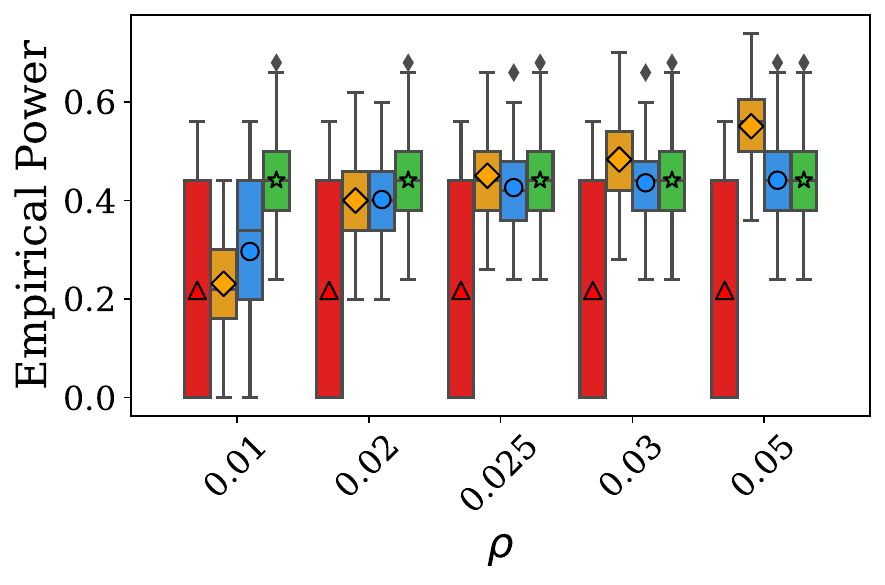}
    \includegraphics[width=0.35\linewidth, valign=c]{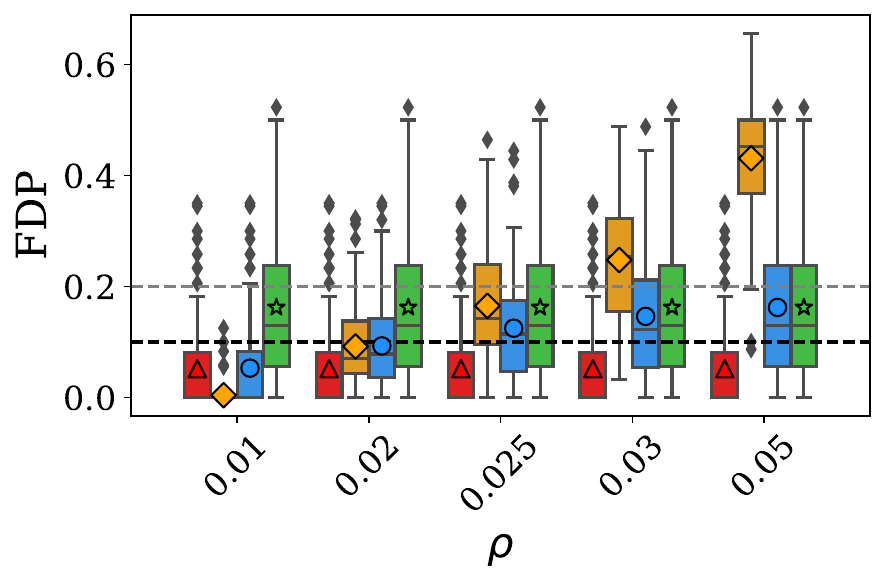}
    \includegraphics[width=0.15\linewidth, valign=c]{figures/experiments/new_outlier_detection/legend_box.pdf}
    \caption{Comparison of conformal outlier detection methods on the \textit{Shuttle} dataset as a function of the trimming proportion $\rho$. Left: detection rate (empirical power). Right: false discovery proportion.}
    \label{app-fig:od-p_trim}
\end{figure}
As $\rho$ increases, \texttt{SynthBH} becomes less conservative, achieving higher power than \texttt{BH~(real)} while achieving FDR close to the nominal level $\alpha$. For high trimming values, \texttt{BH~(synth)} exhibits uncontrolled FDR levels, whereas \texttt{SynthBH} continues to control the FDR at level $\alpha+\eps$, as guaranteed by~\Cref{thm:outlier_prds}. These results illustrate that \texttt{SynthBH} adapts to the unknown quality of the synthetic data, leveraging them when informative and remaining safe when they are not.

\Cref{app-fig:od-fdr-power-e-0.05} shows the same experiment as \Cref{fig:od-fdr-power}, but with a lower value of $\eps=5\%$. The results follow the same trend: when the synthetic data are informative, \texttt{SynthBH} potentially achieve higher power than \texttt{BH~(real)} at the same empirical FDR. When the synthetic data are uninformative, \texttt{SynthBH} matches the performance of \texttt{BH~(real)}. The potential improvement depends on $\eps$, with larger values of $\eps$ generally allowing for greater power gains.

\begin{figure}[!h]
    \centering
    \begin{subfigure}[t]{0.28\textwidth}
    \includegraphics[width=\linewidth, valign=t]{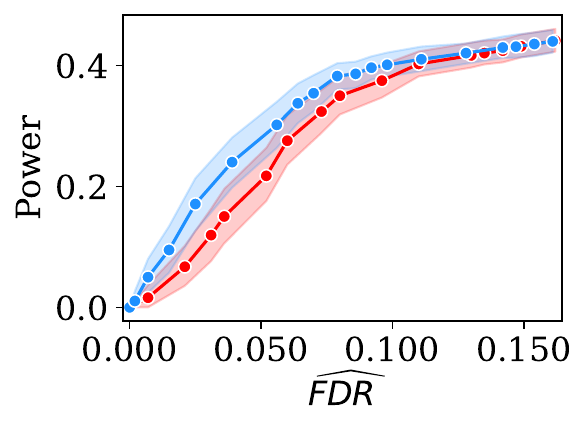}
    \includegraphics[width=\linewidth, valign=t]{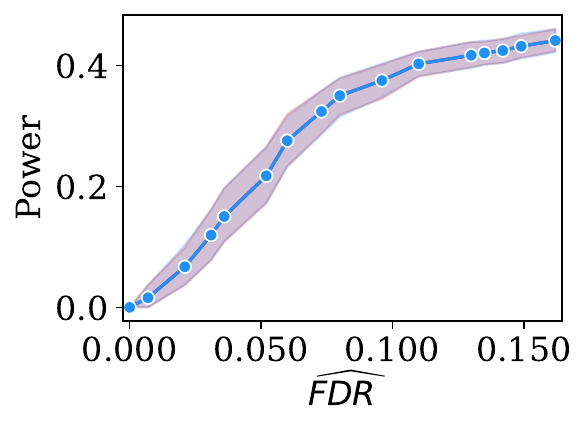}
    \caption{Shuttle}
    \end{subfigure}
    \begin{subfigure}[t]{0.28\textwidth}
    \includegraphics[width=\linewidth, valign=t]{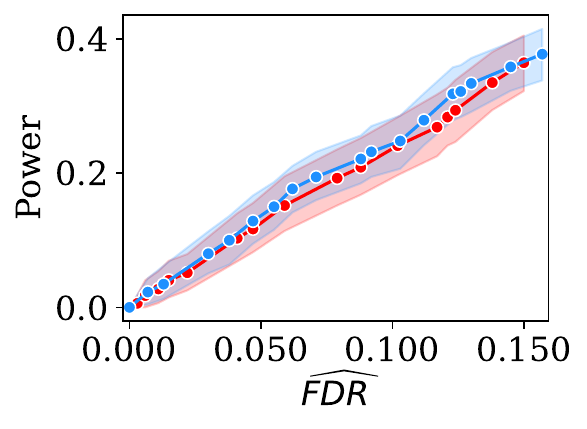}
    \includegraphics[width=\linewidth, valign=t]{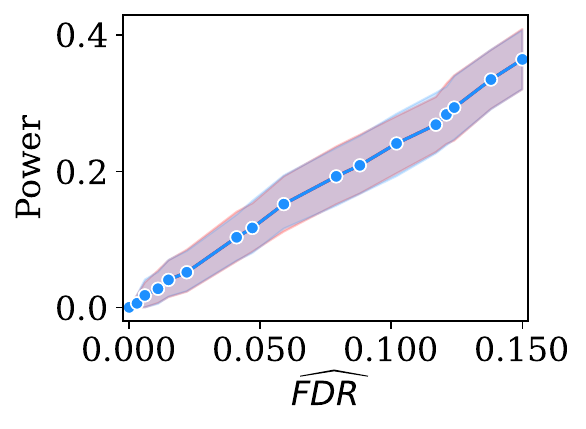}
    \caption{Credit-card}
    \end{subfigure}
    \begin{subfigure}[t]{0.28\textwidth}
    \includegraphics[width=\linewidth, valign=t]{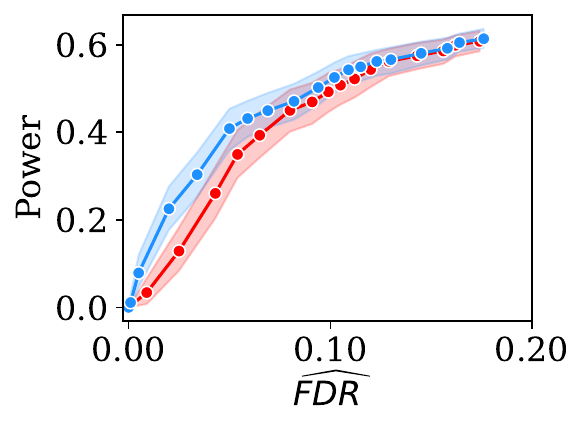}
    \includegraphics[width=\linewidth, valign=t]{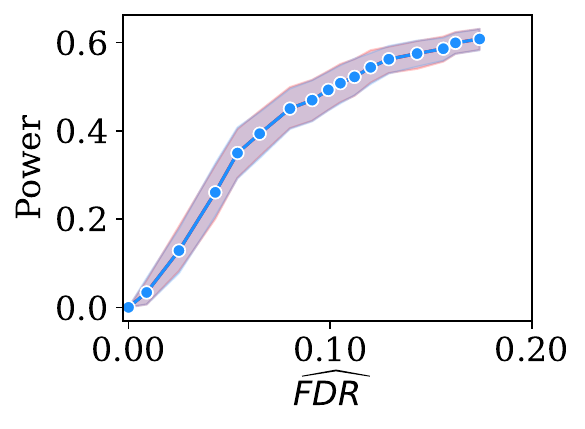}
    \caption{KDDCup99}
    \end{subfigure}
    \includegraphics[width=0.13\linewidth, valign=t]{figures/experiments/new_outlier_detection/legend-roc.pdf}
    \caption{Power as a function of the empirical FDR for \texttt{BH~(real)} and our proposed method \texttt{SynthBH} across different outlier detection datasets: (a) Shuttle, (b) Credit-card, and (c) KDDCup99. The first row corresponds to a trimming proportion of 2\%, and the second row corresponds to a trimming proportion of 0\%, representing a scenario where synthetic data are not useful. Here, $\alpha\in [0,20\%]$. Results are shown for $\eps=5\%$.}
    \label{app-fig:od-fdr-power-e-0.05}
\end{figure}

\subsection{Genomics of Drug Sensitivity in Cancer}\label{app-sec:gdsc-exp}

\Cref{app-fig:gdsc} presents additional results on the GDSC genomic data with $\eps=5\%$. The trends are consistent with those observed in the main manuscript: \texttt{SynthBH} obtains more rejections than \texttt{BH~(real)} while maintaining high average ground-truth scores, notably larger than those of \texttt{BH~(real$+\eps$)}. \texttt{BH~(synth)} shows the same performance as in the main, as it does not depend on $\eps$. Compared to the results with $\eps=10\%$ (\Cref{fig:gdsc}), \texttt{SynthBH} produces slightly fewer rejections, illustrating that the gain in power depends on the value of $\eps$.

\begin{figure}[!h]
    \centering
    \includegraphics[height=0.17\linewidth]{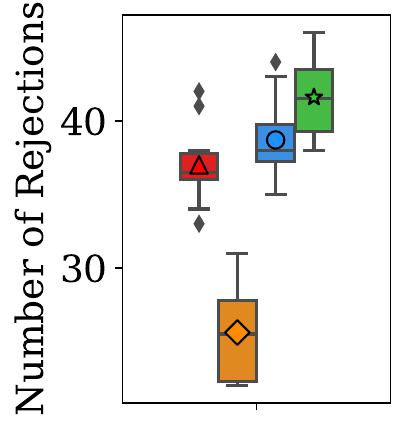}
    \includegraphics[height=0.17\linewidth]{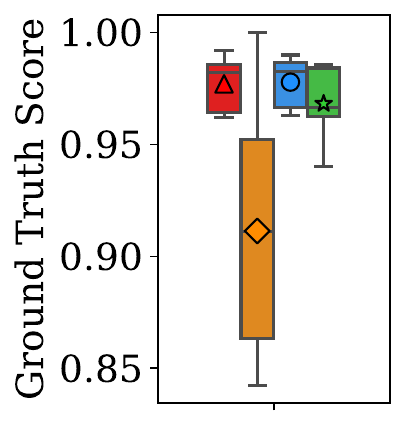}
    \includegraphics[height=0.17\linewidth]{figures/experiments/gdsc/legend.pdf}
    \caption{Comparison of multiple testing methods on GDSC genomic data for 100 tissue-feature-drug hypotheses. Left: number of rejections. Right: average ground-truth score of the rejection set. Results are shown for $\eps=5\%$.}
    \label{app-fig:gdsc}
\end{figure}

\clearpage

\section{Computational Complexity and Efficient Implementation}\label{app-sec:complexity}

\paragraph{Naive implementation cost.}
As written, \Cref{alg:spbh} iterates over $k=1,\ldots,m$ and (i) recomputes the rank-adaptive synthetic-powered p-values $(\tilde p_j^{k\eps/m})_{j\in[m]}$ and (ii) finds the $k$-th order statistic at each iteration.
A direct implementation therefore performs $\Theta(m)$ work per $k$, leading to $\Theta(m^2)$ time in the worst case (and $\mathcal{O}(m)$ memory), which can be prohibitive in large-scale settings.

\paragraph{A static reduction to BH.}
Despite the rank-adaptive definition, \texttt{SynthBH} admits an exact reduction to a single standard BH run on a static set of modified p-values.

Fix $\alpha\in(0,1)$ and $\eps\ge 0$, and define the constant $c:=\alpha/(\alpha+\eps)\in(0,1]$.
For each hypothesis $j\in[m]$, define the static modified p-value
\begin{equation}\label{eqn:static_v}
v_j := p_j \wedge (\tilde p_j \vee (c p_j)).
\end{equation}

\begin{lemma}[Static reduction of \texttt{SynthBH}]\label{lem:static_reduction}
For every $j\in[m]$ and every $k\in[m]$, the step-$k$ acceptance condition in \Cref{alg:spbh} is equivalent to a static BH condition:
$\tilde p_j^{k\eps/m}\le \alpha k/m$ holds if and only if $v_j\le \alpha k/m$.
Consequently, the rejection index $k^*$ defined in \eqref{eqn:rej_num} satisfies
$k^*=\max\{k\in[m]: v_{(k)}\le \alpha k/m\}$, where $v_{(k)}$ is the $k$-th smallest element of $(v_j)_{j\in[m]}$.
Moreover, the rejection set returned by \Cref{alg:spbh} coincides with the BH rejection set at level $\alpha$ applied to $(v_j)_{j\in[m]}$ (with the standard tie convention of rejecting all $j$ such that $v_j\le v_{(k^*)}$).
\end{lemma}

\begin{proof}
Fix $j$ and write $p:=p_j$ and $\tilde p:=\tilde p_j$.
Let $t:=\alpha k/m$ and $\delta:=k\eps/m$, so $\delta=(\eps/\alpha)t$.
By definition, $\tilde p_j^\delta=p\wedge(\tilde p\vee(p-\delta))$.
Thus $\tilde p_j^\delta\le t$ holds if and only if either $p\le t$, or simultaneously $\tilde p\le t$ and $p-\delta\le t$.
The inequality $p-\delta\le t$ is equivalent to $p\le t+\delta=t(1+\eps/\alpha)=(\alpha+\eps)t/\alpha$, which is in turn equivalent to $cp\le t$ for $c=\alpha/(\alpha+\eps)$.
Hence $\tilde p_j^\delta\le t$ is equivalent to $p\le t$ or $(\tilde p\le t \text{ and } cp\le t)$, which is exactly the condition $p\wedge(\tilde p\vee(cp))\le t$, i.e., $v_j\le t$.
Applying this equivalence at $t=\alpha k/m$ for each $k$ yields the claim about the feasible set of $k$s, 
hence the equality of $k^*$, and the equality of the resulting BH-style rejection sets (up to ties).
\end{proof}

\paragraph{Weighted case.}
For weighted \texttt{SynthBH} (\Cref{alg:spbh_weighted}), define $c_j:=\alpha/(\alpha+w_j\eps)$ and
$v_j^{(w)}:=p_j\wedge(\tilde p_j\vee(c_j p_j))$.
The same argument (with $\delta=k w_j\eps/m$) shows that for every $j,k$,
$\tilde p_j^{k w_j\eps/m}\le \alpha k/m$ holds if and only if $v_j^{(w)}\le \alpha k/m$.
Therefore, \Cref{alg:spbh_weighted} is exactly BH at level $\alpha$ applied once to the static values $(v_j^{(w)})_{j\in[m]}$.

\paragraph{Time and memory complexity.}
Computing $(v_j)$ (or $(v_j^{(w)})$) takes $\mathcal{O}(m)$ time and $\mathcal{O}(m)$ memory.
Running BH on these values requires sorting, which takes $\mathcal{O}(m\log m)$ time and $\mathcal{O}(m)$ memory, followed by a linear scan to find $k^*$ and construct the rejection set.
Hence, the step-up overhead of \texttt{SynthBH} matches the classic BH procedure: $\mathcal{O}(m\log m)$ time and $\mathcal{O}(m)$ memory.
In large-scale applications, the dominant additional cost is typically the cost of generating the underlying p-values $(p_j)$ and $(\tilde p_j)$ rather than the step-up rule itself.

\clearpage

\section{Proofs}

\subsection{Proof of Lemma~\ref{lem:ext}}
Let $h(x) = \PPst{X \in A}{W=x}$. Then for $x_1 < x_2$,
\begin{align*}
    &\PPst{X \in A}{W \leq x_2} = \frac{\PP{X \in A, W \leq x_2}}{\PP{W \leq x_2}} = \frac{\EE{\PPst{X \in A, W \leq x_2}{W}}}{\PP{W \leq x_2}}\\
    &= \frac{\EE{\PPst{X \in A}{W}\One{W \leq x_2}}}{\PP{W \leq x_2}} = \frac{\EE{h(W)\One{W \leq x_2}}}{\PP{W \leq x_2}}\\
    &= \frac{\EE{h(W)\One{W \leq x_1}} + \EE{h(W) \One{x_1 < W \leq x_2}}}{\PP{W \leq x_2}}\\
    &= \frac{\EEst{h(W)}{W \leq x_1} \PP{W \leq x_1} + \EEst{h(W)}{x_1 < W \leq x_2} \PP{x_1 < W \leq x_2}}{\PP{W \leq x_2}}\\
    &\geq \frac{\EEst{h(W)}{W \leq x_1} \PP{W \leq x_1} + \EEst{h(W)}{W \leq x_1} \PP{x_1 < W \leq x_2}}{\PP{W \leq x_2}}\\
    &= \EEst{h(W)}{W \leq x_1} = \PPst{X \in A}{W \leq x_1},
\end{align*}
where the inequality holds since $h$ is nondecreasing. This completes the proof.

\subsection{Proof of Theorem~\ref{thm:spbh}}

The proof applies ideas presented in the expository work by~\citet{wang2022elementary}. 

Throughout, let $\delta_r := r\eps/m$ for $r\in[m]$.
Recall that the algorithm outputs an index $k^*\in\{0,1,\dots,m\}$ and rejects the set
\[
\mathcal{R}
=
\left\{j\in[m]: \tilde p_j^{\delta_{k^*}} \le \tilde p_{(k^*)}^{\delta_{k^*}}\right\},
\]
with the convention that $k^*=0$ implies $\mathcal{R}=\emptyset$.
(If ties occur, then $|\mathcal{R}|\ge k^*$; this only makes the bound below more conservative since we upper bound by $k^*$.) 
Write $I_0=\{j\in[m]: H_j \text{ is true}\}$ and $m_0=|I_0|$.

\paragraph{Rewriting FDR and comparing to the BH threshold.}
By definition of $\FDR$ and the rejection rule,
\begin{align*}
\FDR 
&= \mathbb{E}\!\left[\frac{|\mathcal{R}\cap I_0|}{|\mathcal{R}|\vee 1}\right]
\;\le\;
\mathbb{E}\!\left[\frac{\sum_{j\in I_0}\One{\tilde p_j^{\delta_{k^*}}\le \tilde p_{(k^*)}^{\delta_{k^*}}}}{k^*\vee 1}\cdot \One{k^*\ge 1}\right],
\end{align*}
where the inequality uses $|\mathcal{R}|\ge k^*$ (hence $1/(|\mathcal{R}|\vee 1)\le 1/(k^*\vee 1)$) and $\mathcal{R}=\emptyset$ when $k^*=0$.

Next, on the event $\{k^*\ge 1\}$, the definition of $k^*$ (\Cref{alg:spbh} / \eqref{eqn:rej_num}) implies
$\tilde p_{(k^*)}^{\delta_{k^*}}\le \alpha k^*/m$.
Therefore, by monotonicity of indicators (if $a\le b$ then $\One{x\le a}\le \One{x\le b}$),
\(
\One{\tilde p_j^{\delta_{k^*}}\le \tilde p_{(k^*)}^{\delta_{k^*}}}
\le
\One{\tilde p_j^{\delta_{k^*}}\le \alpha k^*/m},
\)
and hence
\begin{align*}
\FDR
&\le
\mathbb{E}\!\left[\frac{\sum_{j\in I_0}\One{\tilde p_j^{\delta_{k^*}}\le \alpha k^*/m}}{k^*\vee 1}\cdot \One{k^*\ge 1}\right].
\end{align*}

\paragraph{Conditioning on the value of $k^*$.}
Use the partition $\One{k^*\ge 1}=\sum_{r=1}^m \One{k^*=r}$ and the fact that on $\{k^*=r\}$ we have $\delta_{k^*}=\delta_r$:
\begin{align*}
\FDR
&\le
\mathbb{E}\!\left[\sum_{r=1}^m \frac{1}{r}\Big(\sum_{j\in I_0}\One{\tilde p_j^{\delta_r}\le \alpha r/m}\Big)\One{k^*=r}\right]
\\
&=
\sum_{r=1}^m\sum_{j\in I_0}\frac{1}{r}\mathbb{E}\!\left[\One{\tilde p_j^{\delta_r}\le \alpha r/m}\One{k^*=r}\right]
\qquad\text{(linearity of expectation)}
\\
&=
\sum_{j\in I_0}\sum_{r=1}^m \frac{1}{r}\mathbb{P}\!\left(\tilde p_j^{\delta_r}\le \frac{\alpha r}{m},\, k^*=r\right)
\qquad\text{(since $\mathbb{E}[\One{A}\One{B}]=\mathbb{P}(A\cap B)$)}.
\end{align*}

\paragraph{Using the deterministic guardrail.}
We have, deterministically for all $j$ and $r$,
\(
\tilde p_j^{\delta_r}\ge p_j-\delta_r = p_j-\frac{r\eps}{m}.
\)
Hence, whenever $\tilde p_j^{\delta_r}\le \alpha r/m$ occurs, we must have
$p_j-\frac{r\eps}{m}\le \frac{\alpha r}{m}$, i.e.,
$p_j\le \frac{r(\alpha+\eps)}{m}$.
Equivalently,
\(
\left\{\tilde p_j^{\delta_r}\le \frac{\alpha r}{m}\right\}
\subseteq
\left\{p_j\le \frac{r(\alpha+\eps)}{m}\right\}.
\)
Thus, 
\begin{align*}
\FDR
&\le
\sum_{j\in I_0}\sum_{r=1}^m \frac{1}{r}\mathbb{P}\!\left(p_j\le \frac{r(\alpha+\eps)}{m},\, k^*=r\right).
\end{align*}

\paragraph{Using super-uniformity of null p-values.}
For each $(j,r)$,
\[
\mathbb{P}\!\left(p_j\le \frac{r(\alpha+\eps)}{m},\, k^*=r\right)
=
\mathbb{P}\!\left(p_j\le \frac{r(\alpha+\eps)}{m}\right)\cdot 
\mathbb{P}\!\left(k^*=r \,\middle|\, p_j\le \frac{r(\alpha+\eps)}{m}\right)
\]
(definition of conditional probability).
Since $j\in I_0$, $p_j$ is super-uniform, so for any $t\in[0,1]$,
$\mathbb{P}(p_j\le t)\le t$.
Applying this with $t=r(\alpha+\eps)/m$ yields
\[
\mathbb{P}\!\left(p_j\le \frac{r(\alpha+\eps)}{m}\right)\le \frac{r(\alpha+\eps)}{m}.
\]
Therefore,
\begin{align*}
\FDR
&\le
\sum_{j\in I_0}\sum_{r=1}^m \frac{1}{r}\cdot \frac{r(\alpha+\eps)}{m}\,
\mathbb{P}\!\left(k^*=r \,\middle|\, p_j\le \frac{r(\alpha+\eps)}{m}\right)
\\
&=
\frac{\alpha+\eps}{m}\sum_{j\in I_0}\sum_{r=1}^m 
\mathbb{P}\!\left(k^*=r \,\middle|\, p_j\le \frac{r(\alpha+\eps)}{m}\right).
\end{align*}

\paragraph{PRDS/telescoping bound.}
Fix $j\in I_0$.
Define $x_r:=r(\alpha+\eps)/m$ and
\[
q_{j,r}:=\mathbb{P}\!\left(k^*\le r \,\middle|\, p_j\le x_r\right),
\qquad r\in[m],
\]
and let $[r]_0:=\{0\}\cup[r]$.
For notational convenience, set $q_{j,0}:=0$.
Whenever conditioning events have probability zero (e.g., $\PP(p_j\le x_r)=0$), we interpret conditional probabilities via an arbitrary version (equivalently, set them to $0$); this does not affect the argument since the corresponding unconditional terms vanish.

\emph{Monotonicity of the event $\{k^*\le r\}$.}
View $k^*$ as a function of $\mathbf p=(p_1,\dots,p_m,\tilde p_1,\dots,\tilde p_m)\in[0,1]^{2m}$, i.e., $k^*=k^*(\mathbf p)$.
By Lemma~\ref{lem:k_dec}, $k^*(\mathbf p)$ is coordinatewise nonincreasing: if $\mathbf p\preceq \mathbf q$ then $k^*(\mathbf p)\ge k^*(\mathbf q)$.
Hence the set
\[
B_r:=\{\mathbf p: k^*(\mathbf p)\le r\} = (k^*)^{-1}([r]_0)
\]
is an \emph{increasing} set: if $\mathbf p\in B_r$ and $\mathbf q\succeq \mathbf p$, then $k^*(\mathbf q)\le k^*(\mathbf p)\le r$, so $\mathbf q\in B_r$.

\emph{Applying PRDS.}
By the PRDS assumption (Definition~\ref{def:prds_ext}) with respect to $(p_1,\dots,p_m)$, for each null $j$ and for any increasing set $A$,
the function $x\mapsto \mathbb{P}(\mathbf P\in A\mid p_j\le x)$ is nondecreasing.
Applying this with $A=B_r$ shows that $x\mapsto \mathbb{P}(k^*\le r\mid p_j\le x)$ is nondecreasing.
Since $x_r\le x_{r+1}$, we obtain, for all $r\in[m-1]$,
\[
q_{j,r}=\mathbb{P}(k^*\le r\mid p_j\le x_r)
\le
\mathbb{P}(k^*\le r\mid p_j\le x_{r+1}).
\]
Equivalently, for $r\ge 1$,
\[
\mathbb{P}(k^*\le r-1\mid p_j\le x_r)\ge \mathbb{P}(k^*\le r-1\mid p_j\le x_{r-1}) = q_{j,r-1}.
\]

\emph{Telescoping.}
Using $\{k^*=r\}=\{k^*\le r\}\setminus\{k^*\le r-1\}$,
\begin{align*}
\mathbb{P}\!\left(k^*=r \,\middle|\, p_j\le x_r\right)
&=
\mathbb{P}\!\left(k^*\le r \,\middle|\, p_j\le x_r\right)
-
\mathbb{P}\!\left(k^*\le r-1 \,\middle|\, p_j\le x_r\right)
\\
&\le
\mathbb{P}\!\left(k^*\le r \,\middle|\, p_j\le x_r\right)
-
\mathbb{P}\!\left(k^*\le r-1 \,\middle|\, p_j\le x_{r-1}\right)
\\
&=
q_{j,r}-q_{j,r-1}.
\end{align*}
Summing over $r=1,\dots,m$ and telescoping gives
\[
\sum_{r=1}^m \mathbb{P}\!\left(k^*=r \,\middle|\, p_j\le x_r\right)
\le
\sum_{r=1}^m (q_{j,r}-q_{j,r-1})
=
q_{j,m}-q_{j,0}
\le 1,
\]
since $q_{j,m}\le 1$ and $q_{j,0}\ge 0$.
Applying the last display for each $j\in I_0$ yields
\[
\FDR
\le
\frac{\alpha+\eps}{m}\sum_{j\in I_0} 1
=
\frac{m_0}{m}(\alpha+\eps),
\]
which proves the claim for \Cref{alg:spbh}.

For \Cref{alg:spbh_weighted}, define $\delta_{j,r}:= r w_j\eps/m$ for $r\in[m]$.
Repeating the same steps as in the unweighted case yields
\begin{align*}
\FDR
&\le
\sum_{j\in I_0}\sum_{r=1}^m \frac{1}{r}\mathbb{P}\left(\tilde p_j^{\delta_{j,r}}\le \frac{\alpha r}{m},\, k^*=r\right)
\\
&\le
\sum_{j\in I_0}\sum_{r=1}^m \frac{1}{r}\mathbb{P}\left(p_j\le \frac{r(\alpha+w_j\eps)}{m},\, k^*=r\right),
\end{align*}
where we used the deterministic inequality $\tilde p_j^{\delta_{j,r}}\ge p_j-\delta_{j,r}=p_j-rw_j\eps/m$.

Now fix $j\in I_0$ and write $x_{j,r}:=r(\alpha+w_j\eps)/m$. Since $p_j$ is super-uniform under $H_j$,
$\PP(p_j\le x_{j,r})\le x_{j,r}$, hence
\begin{align*}
\FDR
&\le
\sum_{j\in I_0}\sum_{r=1}^m \frac{1}{r}\cdot \frac{r(\alpha+w_j\eps)}{m}\,
\PPst{k^*=r}{p_j\le x_{j,r}}
\\
&=
\sum_{j\in I_0}\frac{\alpha+w_j\eps}{m}\sum_{r=1}^m \PPst{k^*=r}{p_j\le x_{j,r}}.
\end{align*}

It remains to bound $\sum_{r=1}^m \PPst{k^*=r}{p_j\le x_{j,r}}$.
Define $q_{j,r}:=\PPst{k^*\le r}{p_j\le x_{j,r}}$ for $r\in[m]$, and set $q_{j,0}:=0$.

By Lemma~\ref{lem:k_dec_weighted}, the set $\{k^*\le r\}$ is increasing in $(p_1,\ldots,p_m,\tilde p_1,\ldots,\tilde p_m)$.
By the PRDS-with-respect-to-$(p_1,\ldots,p_m)$ assumption,
$x\mapsto \PP(k^*\le r\mid p_j\le x)$ is nondecreasing.
Since $x_{j,r}$ is nondecreasing in $r$, the same telescoping argument gives
\(
\sum_{r=1}^m \PPst{k^*=r}{p_j\le x_{j,r}} \le 1.
\)
Therefore,
\begin{align*}
\FDR
&\le
\sum_{j\in I_0}\frac{\alpha+w_j\eps}{m}
=
\frac{m_0}{m}\alpha+\frac{\eps}{m}\sum_{j\in I_0}w_j
\le
\frac{m_0}{m}\alpha+\eps,
\end{align*}
where the last inequality uses $\sum_{j=1}^m w_j=m$.

The following lemmas are used in the proof. 
\begin{lemma}\label{lem:k_dec}
    Let $k^* = k^*(\mathbf{p})$ be defined as in~\eqref{eqn:rej_num}, where $\mathbf{p} = (p_1, p_2, \ldots, p_m, \tilde{p}_1, \ldots, \tilde{p}_m) \in [0,1]^{2m}$. Then for any fixed $\mathbf{p}, \mathbf{q} \in [0,1]^{2m}$ such that $\mathbf{p} \preceq \mathbf{q}$, it holds that $k^* (\mathbf{p}) \geq k^*(\mathbf{q})$.
\end{lemma}

\begin{proof}[Proof of Lemma~\ref{lem:k_dec}]
    Fix any $\mathbf{p} \preceq \mathbf{q}$, and write $\mathbf{p} = (p_1, p_2, \ldots, p_m, \tilde{p}_1, \ldots, \tilde{p}_m)$ and $\mathbf{q} = (q_1, q_2, \ldots, q_m, \tilde{q}_1, \ldots, \tilde{q}_m)$.
    It suffices to show that
    \begin{equation}\label{eqn:sub}
        \left\{k \in [m] : \frac{m \cdot \tilde{q}_{(k)}^{k\eps/m}}{k} \leq \alpha\right\} \subset \left\{k \in [m] : \frac{m \cdot \tilde{p}_{(k)}^{k\eps/m}}{k} \leq \alpha\right\},
    \end{equation}
    where $\tilde{q}_j^\delta$ and $\tilde{q}_{(k)}^{k\eps/m}$ are defined analogously to $\tilde{p}_j^\delta$ and $\tilde{p}_{(k)}^{k\eps/m}$, respectively.

    Observe that, by the definition of $\tilde{p}_j^\delta$ and $\tilde{q}_j^\delta$, $\mathbf{p} \preceq \mathbf{q}$ implies $\tilde{p}_j^\delta \leq \tilde{q}_j^\delta$ for any $j \in [m]$ and $\delta \ge 0$, which in turn implies $\tilde{p}_{(k)}^\delta \leq \tilde{q}_{(k)}^\delta$ for any $k \in [m]$ and $\delta \ge 0$. This proves the subset relation~\eqref{eqn:sub}.

\end{proof}

\begin{lemma}\label{lem:k_dec_weighted}
Fix weights $(w_j)_{j\in[m]}$ and let $k^*_w = k^*_w(\mathbf p)$ be the rejection index produced by \Cref{alg:spbh_weighted}, viewed as a function of $\mathbf p=(p_1,\ldots,p_m,\tilde p_1,\ldots,\tilde p_m)\in[0,1]^{2m}$.
Then $k^*_w$ is coordinatewise nonincreasing: if $\mathbf p\preceq \mathbf q$, then $k^*_w(\mathbf p)\ge k^*_w(\mathbf q)$.
\end{lemma}

\begin{proof}
Fix $\mathbf p\preceq \mathbf q$. For each $k\in[m]$ and each $j\in[m]$, define
$\tilde p_{j}^{(k)} := p_j \wedge (\tilde p_j \vee (p_j - k w_j\eps/m))$
and define $\tilde q_{j}^{(k)}$ analogously from $\mathbf q$.
Because $\min(\cdot,\cdot)$ and $\max(\cdot,\cdot)$ are coordinatewise nondecreasing, $\mathbf p\preceq \mathbf q$ implies $\tilde p_{j}^{(k)}\le \tilde q_{j}^{(k)}$ for all $j$, hence the $k$-th order statistics satisfy $\tilde p_{(k)}^{(k)}\le \tilde q_{(k)}^{(k)}$.
Therefore any $k$ satisfying $\tilde q_{(k)}^{(k)}\le \alpha k/m$ also satisfies $\tilde p_{(k)}^{(k)}\le \alpha k/m$, implying the feasible set of $k$'s for $\mathbf q$ is contained in that for $\mathbf p$ and hence $k^*_w(\mathbf p)\ge k^*_w(\mathbf q)$.
\end{proof}

\subsection{Proof of Theorem~\ref{thm:outlier_prds}}
The proof applies 
ideas developed in~\citet{bates2023testing}. Without loss of generality, it is enough to assume $1 \in I_0$ and show that for any increasing set $A \subset [0,1]^{2m}$,
\[\PPst{(p_1, p_2, \ldots, p_m, \tilde{p}_1, \ldots, \tilde{p}_m) \in A}{p_1 = t}\]
is a nondecreasing function of $t$. 

For simplicity, let us write $S_i = s(X_i)$ for $i \in [n+m]$ and $\tilde{S}_i = s(\tilde{X}_i)$ for $i \in [N]$. Observe that $p_1 = \frac{R}{n+1}$, where
\[R = \sum_{i=1}^n \One{S_i \geq S_{n+1}} + 1 = \sum_{i=1}^n \One{S_{(i)} \geq S_{n+1}} + 1\]
denotes the rank of $S_{n+1}$ among $(S_1, \ldots, S_n, S_{n+1})$, and $S_{(1)} > S_{(2)} > \ldots > S_{(n)}$ denote the order statistics of $S_1, \ldots, S_n$ in a decreasing order.
Therefore, it suffices to show that for any $1 \leq r_1 < r_2 \leq n+1$,
\[\PPst{(p_1, p_2, \ldots, p_m, \tilde{p}_1, \ldots, \tilde{p}_m) \in A}{R = r_1} \leq \PPst{(p_1, p_2, \ldots, p_m, \tilde{p}_1, \ldots, \tilde{p}_m) \in A}{R = r_2}.\]
Let $\bar{S}_{(1)} > \bar{S}_{(2)} > \ldots > \bar{S}_{(n)} > \bar{S}_{(n+1)}$ be the order statistics of $S_1,\ldots, S_n, S_{n+1}$ in a decreasing order. 
Then, by the definition of $R$, we have
\begin{equation}\label{eqn:s_ordered}
    (S_{(1)}, S_{(2)}, \ldots, S_{(n)}) = (\bar{S}_{(1)}, \ldots, \bar{S}_{(R-1)}, \bar{S}_{(R+1)}, \ldots, \bar{S}_{(n+1)}).
\end{equation}
Now fix any $r\in[n+1]$. Define a rank-$r$ hypothetical vector of p-values as follows.
Set
$p_1^r := r/(n+1)$ and
\[
\tilde p_1^r
:=
\frac{r+\sum_{i=1}^N \One{\tilde S_i \ge \bar S_{(r)}}}{n+N+1},
\]
and for each $j=2,\ldots,m$, define
\[
p_j^r
:=
\frac{1+\sum_{i=1,\,i\neq r}^{n+1}\One{\bar S_{(i)} \ge S_{n+j}}}{n+1},
\qquad
\tilde p_j^r
:=
\frac{1+\sum_{i=1,\,i\neq r}^{n+1}\One{\bar S_{(i)} \ge S_{n+j}}
+\sum_{i=1}^N \One{\tilde S_i \ge S_{n+j}}}{n+N+1}.
\]
By construction, on the event $\{R=r\}$ we have $(p_1,\ldots,p_m,\tilde p_1,\ldots,\tilde p_m)=(p_1^r,\ldots,p_m^r,\tilde p_1^r,\ldots,\tilde p_m^r)$.
Observe also that $p_j^{r_1} \leq p_j^{r_2}$ and $\tilde{p}_j^{r_1} \leq \tilde{p}_j^{r_2}$ hold for any $r_1 < r_2$, by construction. 
Indeed, for $j\ge 2$ the only dependence on $r$ in $p_j^r$ (and $\tilde p_j^r$) is through omitting $\bar S_{(r)}$ from the sum over $\bar S_{(i)}$.
Since $r_1<r_2$ implies $\bar S_{(r_1)}>\bar S_{(r_2)}$, removing $\bar S_{(r_1)}$ can only decrease (relative to removing $\bar S_{(r_2)}$) the count of indices $i$ for which $\bar S_{(i)}\ge S_{n+j}$, hence $p_j^{r_1}\le p_j^{r_2}$ and similarly for $\tilde p_j^r$.
Therefore,
\begin{align*}
    &\PPst{(p_1, p_2, \ldots, p_m, \tilde{p}_1, \ldots, \tilde{p}_m) \in A}{R = r_1} = \PPst{(p_1^R, p_2^R, \ldots, p_m^R, \tilde{p}_1^R, \ldots, \tilde{p}_m^R) \in A}{R = r_1}\\
    &= \PPst{(r_1/(n+1), p_2^{r_1}, \ldots, p_m^{r_1}, \tilde{p}_1^{r_1}, \ldots, \tilde{p}_m^{r_1}) \in A}{R = r_1} = \PP{(r_1/(n+1), p_2^{r_1}, \ldots, p_m^{r_1}, \tilde{p}_1^{r_1}, \ldots, \tilde{p}_m^{r_1}) \in A},
\end{align*}
where the last equality holds since for each fixed $r$ the vector $(p_1^r,\ldots,p_m^r,\tilde p_1^r,\ldots,\tilde p_m^r)$ is a measurable function of $(\bar S_{(1)},\ldots,\bar S_{(n+1)})$, $(\tilde S_1,\ldots,\tilde S_N)$, and $(S_{n+2},\ldots,S_{n+m})$, all of which are independent of $R$ under $H_1$ and the stated independence assumptions.
Indeed, under $H_1$ and the distinctness assumption, the rank $R$ of $S_{n+1}$ among $(S_1,\ldots,S_n,S_{n+1})$ is uniform on $[n+1]$ and independent of the order statistics $(\bar S_{(1)},\ldots,\bar S_{(n+1)})$.
 Moreover,
 since $A$ is an increasing set, and
 due to the inequalities observed above,
 this is upper bounded by 
\begin{align*}
    & \PP{(r_2/(n+1), p_2^{r_2}, \ldots, p_m^{r_2}, \tilde{p}_1^{r_2}, \ldots, \tilde{p}_m^{r_2}) \in A} = \PPst{(p_1, p_2, \ldots, p_m, \tilde{p}_1, \ldots, \tilde{p}_m) \in A}{R = r_2},
\end{align*}
Thus, by Lemma \ref{lem:ext}, we conclude that
the PRDS property from Definition \ref{def:prds_ext} holds.
By Theorem \ref{thm:spbh},
this finishes the proof. 

 \clearpage
 
\section{Additional Experimental Details}
\subsection{Outlier Detection Experimental Details}\label{app-sec:datasets}
\paragraph{Experimental setup and metrics.} Unless stated otherwise, we use $n=500$ inlier samples $\Dn$, and $N=2,500$ contaminated samples $\tDn$, containing 5\% outliers, with trimming proportion of $\rho=2\%$.
The outlier detection model is an Isolation Forest \citep{liu2008isolation} with 100 estimators, trained on $5,000$ samples ($2,500$ for the Credit-card dataset) with the same contamination rate as $\tDn$. 
The test set contains $1,000$ 
datapoints with 5\% outliers. The target FDR level is $\alpha=10\%$, and $\eps=10\%$.
All sets ($\Dn$, $\tDn$, the training set, and the test set) are disjoint.
We report the detection rate and the false discovery proportion over 100 random splits of the data.

\paragraph{Dataset details.}
We consider three tabular benchmark datasets, accessed via the \texttt{openML} Python package \citep{OpenML2021,OpenML2021p}. Below, we summarize the key characteristics of each dataset, including the total number of samples, the number of outliers, and the number of features.

\begin{itemize}
    \item \textbf{Shuttle:} This dataset involves the task of classifying space shuttle flight data, where a subset of classes are considered outliers. It contains $58{,}000$ samples, with $12{,}414$ labeled as outliers, and has $9$ features.
    
    \item \textbf{KDDCup99:} Used for network intrusion detection, we focus on detecting \emph{probe} attacks. The dataset contains $101{,}384$ samples, of which $4{,}107$ correspond to probe attacks (outliers), including the types \texttt{Satan}, \texttt{Ipsweep}, \texttt{Nmap}, and \texttt{Portsweep}. The data has $41$ features.
    
    \item \textbf{Credit Card Fraud:} This dataset consists of European credit card transactions, with the task of identifying fraudulent activity. It contains $284,807$ samples, including $492$ fraudulent transactions (treated as outliers), with $29$ features.
\end{itemize}

\subsection{Genomics of Drug Sensitivity in Cancer}\label{app-sec:gdsc-details}
We provide the full list of hypotheses tested in the GDSC experiments. As described in \Cref{sec:exp-gdsc}, each hypothesis corresponds to a specific tissue–feature–drug triple. 
In total, we test 100 hypotheses. The tested tissue–feature–drug triples are listed below.

\begin{longtable}{lll}
\caption{List of tissue-feature-drug hypotheses tested in the GDSC experiments. BRCA denotes breast invasive carcinoma and LUAD denotes lung adenocarcinoma.}\label{tab:gdsct-hypotheses} \\
\toprule
 Tissue & Genomic feature & Drug \\
\midrule
\endfirsthead
\toprule
 Tissue & Genomic feature & Drug \\
\midrule
\endhead
\midrule
\multicolumn{3}{r}{Continued on next page}
\endfoot
\bottomrule
\endlastfoot

BRCA	& gain\_cnaPANCAN301\_(CDK12,ERBB2,MED24)	& Ibrutinib \\
BRCA	& PIK3CA\_mut	& Taselisib \\
BRCA	& gain\_cnaPANCAN301\_(CDK12,ERBB2,MED24)	& Sapitinib \\
BRCA	& gain\_cnaPANCAN301\_(CDK12,ERBB2,MED24)	& Osimertinib \\
BRCA	& PIK3CA\_mut	& Alpelisib \\
BRCA	& gain\_cnaPANCAN301\_(CDK12,ERBB2,MED24)	& Afatinib \\
BRCA	& gain\_cnaPANCAN301\_(CDK12,ERBB2,MED24)	& Lapatinib \\
BRCA	& gain\_cnaPANCAN87	& Cediranib \\
BRCA	& PTEN\_mut	& AZD8186 \\
BRCA	& loss\_cnaPANCAN93	& VTP-B \\
BRCA	& loss\_cnaPANCAN94	& VTP-B \\
BRCA	& loss\_cnaPANCAN95	& VTP-B \\
BRCA	& loss\_cnaPANCAN96	& VTP-B \\
BRCA	& PTEN\_mut	& Afuresertib \\
LUAD	& loss\_cnaPANCAN92	& Daporinad \\
BRCA	& gain\_cnaPANCAN301\_(CDK12,ERBB2,MED24)	& Alpelisib \\
BRCA	& PTEN\_mut	& Ipatasertib \\
BRCA	& gain\_cnaPANCAN87	& PFI-1 \\
LUAD	& STK11\_mut	& Nutlin-3a (-) \\
BRCA	& TP53\_mut	& Nutlin-3a (-) \\
LUAD	& loss\_cnaPANCAN96	& Sapitinib \\
BRCA	& gain\_cnaPANCAN88\_(PABPC1,UBR5)	& Cediranib \\
BRCA	& PTEN\_mut	& AZD5363 \\
BRCA	& loss\_cnaPANCAN210\_(FAT1,IRF2)	& Uprosertib \\
BRCA	& gain\_cnaPANCAN87	& Foretinib \\
BRCA	& gain\_cnaPANCAN367\_(ARFGAP1,GNAS)	& Tozasertib \\
LUAD	& loss\_cnaPANCAN96	& AZD3759 \\
BRCA	& gain\_cnaPANCAN367\_(ARFGAP1,GNAS)	& Pyridostatin \\
BRCA	& gain\_cnaPANCAN272\_(PIP5K1A,SETDB1)	& Navitoclax \\
BRCA	& gain\_cnaPANCAN280\_(ARID4B,FH)	& BMS-754807 \\
BRCA	& gain\_cnaPANCAN88\_(PABPC1,UBR5)	& Vinorelbine \\
BRCA	& gain\_cnaPANCAN363\_(ASXL1)	& Sapitinib \\
BRCA	& gain\_cnaPANCAN177\_(IL7R)	& Eg5\_9814 \\
BRCA	& gain\_cnaPANCAN87	& Paclitaxel \\
BRCA	& loss\_cnaPANCAN210\_(FAT1,IRF2)	& PF-4708671 \\
LUAD	& loss\_cnaPANCAN96	& Erlotinib \\
BRCA	& gain\_cnaPANCAN88\_(PABPC1,UBR5)	& Paclitaxel \\
LUAD	& gain\_cnaPANCAN164\_(KRAS)	& TAF1\_5496 \\
LUAD	& loss\_cnaPANCAN96	& Gefitinib \\
LUAD	& STK11\_mut	& Ipatasertib \\
LUAD	& TP53\_mut	& ERK\_2440 \\
BRCA	& gain\_cnaPANCAN302\_(CLTC,PPM1D)	& AZD7762 \\
LUAD	& gain\_cnaPANCAN164\_(KRAS)	& OF-1 \\
BRCA	& gain\_cnaPANCAN280\_(ARID4B,FH)	& MK-1775 \\
BRCA	& gain\_cnaPANCAN124\_(EGFR)	& Gefitinib \\
LUAD	& loss\_cnaPANCAN94	& Sapitinib \\
LUAD	& loss\_cnaPANCAN95	& Sapitinib \\
LUAD	& gain\_cnaPANCAN164\_(KRAS)	& LGK974 \\
LUAD	& gain\_cnaPANCAN91\_(MYC)	& KRAS (G12C) Inhibitor-12 \\
LUAD	& loss\_cnaPANCAN96	& Ibrutinib \\
LUAD	& loss\_cnaPANCAN337	& Veliparib \\
LUAD	& loss\_cnaPANCAN337	& Serdemetan \\
LUAD	& loss\_cnaPANCAN337	& Sapitinib \\
LUAD	& loss\_cnaPANCAN337	& Lenalidomide \\
LUAD	& loss\_cnaPANCAN337	& Avagacestat \\
LUAD	& loss\_cnaPANCAN337	& VE821 \\
LUAD	& loss\_cnaPANCAN337	& Vismodegib \\
LUAD	& loss\_cnaPANCAN338	& CCT007093 \\
LUAD	& loss\_cnaPANCAN337	& GNE-317 \\
LUAD	& loss\_cnaPANCAN337	& BMS-754807 \\
LUAD	& loss\_cnaPANCAN337	& Motesanib \\
LUAD	& loss\_cnaPANCAN337	& JQ1 \\
LUAD	& loss\_cnaPANCAN337	& Tanespimycin \\
LUAD	& loss\_cnaPANCAN337	& Bosutinib \\
LUAD	& loss\_cnaPANCAN337	& BI-2536 \\
LUAD	& loss\_cnaPANCAN337	& Rucaparib \\
LUAD	& loss\_cnaPANCAN338	& Pyridostatin \\
LUAD	& loss\_cnaPANCAN337	& Temsirolimus \\
LUAD	& loss\_cnaPANCAN338	& AGK2 \\
LUAD	& loss\_cnaPANCAN338	& Olaparib \\
LUAD	& loss\_cnaPANCAN338	& AMG-319 \\
LUAD	& loss\_cnaPANCAN338	& PD0325901 \\
LUAD	& loss\_cnaPANCAN337	& RO-3306 \\
LUAD	& loss\_cnaPANCAN338	& AZD4547 \\
LUAD	& loss\_cnaPANCAN338	& Crizotinib \\
LUAD	& loss\_cnaPANCAN337	& PD173074 \\
LUAD	& loss\_cnaPANCAN337	& Fulvestrant \\
LUAD	& loss\_cnaPANCAN338	& Luminespib \\
LUAD	& loss\_cnaPANCAN338	& SB216763 \\
LUAD	& loss\_cnaPANCAN338	& Ribociclib \\
LUAD	& loss\_cnaPANCAN337	& MK-1775 \\
LUAD	& loss\_cnaPANCAN337	& Doramapimod \\
LUAD	& loss\_cnaPANCAN337	& Dactinomycin \\
LUAD	& loss\_cnaPANCAN337	& Afatinib \\
LUAD	& loss\_cnaPANCAN337	& Talazoparib \\
LUAD	& loss\_cnaPANCAN337	& Crizotinib \\
LUAD	& loss\_cnaPANCAN337	& Buparlisib \\
LUAD	& loss\_cnaPANCAN338	& LJI308 \\
LUAD	& loss\_cnaPANCAN337	& Dasatinib \\
LUAD	& loss\_cnaPANCAN338	& Ulixertinib \\
LUAD	& loss\_cnaPANCAN338	& Ipatasertib \\
LUAD	& loss\_cnaPANCAN338	& Cediranib \\
LUAD	& loss\_cnaPANCAN338	& 5-Fluorouracil \\
LUAD	& loss\_cnaPANCAN337	& LJI308 \\
LUAD	& loss\_cnaPANCAN338	& Vinorelbine \\
LUAD	& loss\_cnaPANCAN337	& Gefitinib \\
LUAD	& loss\_cnaPANCAN338	& Osimertinib \\
LUAD	& loss\_cnaPANCAN338	& RO-3306 \\
LUAD	& loss\_cnaPANCAN337	& Bortezomib \\
LUAD	& loss\_cnaPANCAN337	& 5-Fluorouracil \\
\end{longtable}

\end{document}